\documentclass[10pt,twocolumn,twoside]{IEEEtran}%,draft

\usepackage{amsmath,amssymb,amsthm}
\usepackage{algorithmic}
\usepackage[dvips]{graphicx}
\usepackage{bm}
\usepackage{color}
\usepackage{float}
\newtheorem{theorem}{Theorem}

\newtheorem{lemma}{Lemma}

\newtheorem{proposition}{Proposition}

\newtheorem{definition}{Definition}

\newcommand{\bc}{\mathbf{c}}
\newcommand{\be}{\mathbf{e}}
\newcommand{\bh}{\mathbf{h}}

\newcommand{\bv}{\mathbf{v}}

\newcommand{\bI}{\mathbf{I}}

\newcommand{\bx}{\mathbf{x}}
\newcommand{\bX}{\mathbf{X}}

\newcommand{\bp}{\mathbf{p}}
\newcommand{\br}{\mathbf{r}}

\newcommand{\bA}{\mathbf{A}}
\newcommand{\bB}{\mathbf{B}}

\newcommand{\bH}{\mathbf{H}}
\newcommand{\bM}{\mathbf{M}}

\newcommand{\bQ}{\mathbf{Q}}

\newcommand{\T}{\mathsf{T}}

\begin{document}

\title{An Enhanced SDR based Global Algorithm for Nonconvex Complex Quadratic Programs with Signal Processing Applications\thanks{The work of C. Lu was  supported by the National Natural Science Foundation of China (NSFC) under Grant 11701177 and  Grant 11771243. The work of Y.-F. Liu was supported in part by the NSFC under Grant 11991021,  Grant 11688101,  Grant 11671419, and Grant 11631013. The work of J. Zhou was supported by the NSFC under Grant 11701512.
		Part of this work \cite{luefficienticcc} has been presented (as an invited paper) at the IEEE International Conference on Communications in China (ICCC), Qingdao, China, October 22--24, 2017. (\emph{Corresponding author: Ya-Feng Liu.})}}

\author{ \IEEEauthorblockN{Cheng Lu, Ya-Feng Liu, and Jing Zhou}% <-this % stops a space
\thanks{C. Lu is with the School of Economics and Management, North China Electric Power University, Beijing 102206, China
(e-mail: lucheng1983@163.com). Y.-F. Liu is with the State Key Laboratory
of Scientific and Engineering Computing, Institute of Computational
Mathematics and Scientific/Engineering Computing, Academy of
Mathematics and Systems Science, Chinese Academy of Sciences,
Beijing 100190, China (e-mail: yafliu@lsec.cc.ac.cn). J. Zhou is with the
Department of Applied Mathematics, College of Science, Zhejiang University of Technology, Hangzhou
310023, China (e-mail: zhoujing@zjut.edu.cn).}% <-this % stops a space
}

%\markboth{Journal of \LaTeX\ Class Files,~Vol.~14, No.~8, August~2016}%

\maketitle

\begin{abstract}
In this paper, we consider a class of nonconvex \emph{complex} quadratic programming (CQP) problems, which find a broad spectrum of signal processing applications.
By using the polar coordinate representations of the complex variables, we first derive a new enhanced semidefinite relaxation (SDR) for problem (CQP). %which is tighter than the conventional SDP relaxation.
Based on the newly derived SDR, we further propose an efficient branch-and-bound algorithm for solving problem (CQP).
Key features of our proposed algorithm are: (1) it is guaranteed to find the global solution of the problem (within any given error tolerance); (2) it
is computationally efficient because it carefully utilizes the special structure of the problem. %and the complex variables in it. %(e.g., using their polar coordinate representations).
We apply our proposed algorithm to solve the multi-input multi-output (MIMO) detection problem, the unimodular radar code design problem, and the virtual beamforming design problem. Simulation results show that our proposed enhanced SDR, when applied to the above problems, is generally much tighter than the conventional SDR and our proposed global algorithm can efficiently solve these problems. In particular, our proposed algorithm significantly outperforms the state-of-the-art sphere decode algorithm for solving the MIMO detection problem in the hard cases (where the number of inputs and outputs is equal or the signal-to-noise-ratio is low) and a state-of-the-art general-purpose global optimization solver called Baron for solving the virtual beamforming design problem.
\end{abstract}

\IEEEpeerreviewmaketitle

\section{Introduction}\label{sec:introduction}

In this paper, we consider the following nonconvex complex quadratic programming problem:
\begin{align*}
\min_{\bx \in \mathbb{C}^{n}} ~&~ F(\bx):=\frac{1}{2} \bx^{\dagger}\bQ\bx+\mathrm{Re}\left(\bc^{\dagger}\bx\right) \\ \tag{CQP}
\mbox{s.t.~} ~&~ \ell_i \leq |x_i|\leq u_i,~i=1,2,\ldots,n,\\
&~\arg\left(x_i\right)\in \mathcal{A}_i,~i=1,2,\ldots,n,
\end{align*}
where
\begin{itemize}
  \item [-] $\bx=[x_1,x_2,\ldots,x_n]^\T \in \mathbb{C}^{n}$ is the $n$-dimensional complex (unknown) variable;
  \item [-] $\bQ \in \mathbb{C}^{n\times n}$ is a Hermitian matrix, $\bc \in \mathbb{C}^{n}$ is a complex vector, $u_i$ and $\ell_i$~($i=1,2,\ldots,n$), satisfying $u_i\geq \ell_i\geq 0,$ are $2n$ real numbers, and $\mathcal{A}_i~(i=1,2,\ldots,n)$ are $n$ discrete/continuous sets; and
  \item [-] $\mathrm{Re}(\cdot),~|\cdot|,$ and $\arg\left(\cdot\right)$ denote the real part, the magnitude, and the argument of a complex number, and $(\cdot)^{\T}~\text{and}~(\cdot)^\dagger$ denote the transpose and Hermitian transpose of a (complex) vector.%, respectively.
  \end{itemize}

%\subsection{Motivating applications}
Problem (CQP) finds many important signal processing applications, including
multi-input multi-output (MIMO) detection \cite{Jalden,Ma2004}, unimodular radar code design \cite{Maio2009,Maio2011,Soltanalian}, virtual beamforming \cite{Hong}, phase recovery \cite{Waldspurger}, {sensor bias estimation \cite{pu2018optimal}}, and angular synchronization \cite{Bandeira}; {see Section \ref{sec:application} further ahead.}
%, and phase synchronization \cite{Boumal2016} [[[\textbf{Are angular synchronization and phase synchronization different? If they are similar, let us combine them together.}]]].
For more applications of problem (CQP) in signal processing and communications, please refer to
\cite{Luo2010} and \cite{Palomar} and references therein. In addition, problem (CQP) has also attracted much attention in the mathematical programming community \cite{So2008,Zhang2006,Lu2018,Goemans1995,Goemans2004,JarreJoGO}. For instance, some well-known combinatorial optimization problems, including the max-cut problem \cite{Goemans1995} and the max-3-cut problem \cite{Goemans2004}, and the so-called unit-modulus constrained QP \cite{Lu2018} can all be recast into the form of problem (CQP). %Problem (CQP) also arises in power system applications \cite{Low2014}.

It is known that problem (CQP) is NP-hard in general \cite{Zhang2006}. %(except some special cases \cite{liu11tspbeamforming,yu2007transmitter}).
Hence, there is no polynomial time algorithm which can solve it to global optimality (unless P=NP).
{{Most of existing algorithms for solving problem (CQP) are approximation algorithms, local optimization algorithms, or other heuristics}} (e.g., \cite{Ma2004,Maio2009,Maio2011,Waldspurger,Bandeira,Luo2010,So2008,Zhang2006,Gershman,He,Luo2007}).
{{These algorithms generally cannot guarantee to find the global solutions of problem (CQP), except only for some special cases \cite{pu2018optimal,Bandeira,Ge2015,lu2019tightness}.}}
%However, in general cases, the global optimality of the solutions obtained by these algorithms is not guaranteed.
A straightforward way of globally solving problem (CQP) is to first reformulate the problem as
an equivalent real QP by representing the complex variables by their real and imaginary components and then apply the
existing general-purpose global algorithms (e.g., algorithms proposed in \cite{Linderoth,Tawarmalani}) for solving the equivalent real reformulation.
However, an issue of doing so is the computational efficiency (see our numerical results in Section \ref{sec:numericalbaron}), since it does not utilize the
special structure of the problem with the complex variables. %in it.

%To the best of our knowledge, there is no global algorithm that is specially designed and can efficiently solve problem
%(CQP) by utilizing the special structure of the complex variables. The goal of this paper is to fill this gap, i.e., propose
%an efficient global algorithm for solving problem (CQP). More specifically, we propose an efficient branch-and-bound
%algorithm for globally solving problem (CQP). Our proposed algorithm directly deals with the complex problem by using
%the polar coordinate representations of the complex variables, and branches on the ranges of variables in polar coordinate
%representations. One key feature of working on the polar coordinate representations, compared to working on the real
%and imaginary representations, is that a much tighter SDP relaxation can be obtained. Our preliminary numerical results
%show that our proposed algorithm is able to efficiently solve problem (CQP).
To the best of our knowledge, there is no global algorithm that is specially designed and can efficiently solve problem (CQP) by utilizing the special structure of the problem.
%and the complex variables in it.
The goal of this paper is to fill this gap, i.e., propose an efficient\footnote{{The term ``efficient" in this paper means that the corresponding algorithm is computationally efficient, which does not imply that the algorithm has a polynomial time complexity.}} global algorithm for solving problem (CQP). The main contributions of this paper are twofold.
\begin{itemize}
  \item \emph{A New and Enhanced SDR for Problem (CQP)}. We first give an equivalent reformulation of problem (CQP) by using the polar coordinate representations of the complex variables. The equivalent reformulation reveals the
  intrinsical nonconvexity of problem (CQP). Then, we derive the \emph{convex envelope}\footnote{For a given set, its convex envelope is defined as the {smallest} convex set that contains it.} of these nonconvex constraints in the reformulation and use their convex envelopes to
  replace the original nonconvex constraints, which thus lead to a new SDR for problem (CQP). The new SDR for problem (CQP) is generally (much) tighter than the conventional SDR, which directly drops the nonconvex constraints in the equivalent reformulation. It is worth mentioning that the new enhanced SDR for problem (CQP) is computationally very efficient as all newly added constraints as compared with the conventional SDR are linear constraints. %, because the convex envelope constraints are computationally
%  It is worth mentioning that: (1) the new enhanced SDR for problem (CQP) is computationally very efficient as all newly added constraints as compared with the conventional SDR are linear constraints; (2) the new enhanced SDR for problem (CQP) is a nontrivial extension of the SDRs in \cite{Lu2018,lu2019tightness} and the proposed algorithm in the next section is a nontrivial extension of the algorithms in \cite{Lu2018,lu2019tightness}

  \item \emph{An Efficient Global Algorithm for Problem (CQP)}. Based on the newly derived SDR, we propose an efficient branch-and-bound algorithm for problem (CQP) that is guaranteed to find its global solution (within any given error tolerance). The newly derived SDR plays a very crucial role in the proposed branch-and-bound algorithm, because the algorithm needs to solve an SDR at each iteration and the optimal values of all solved SDRs will provide a lower bound for problem (CQP). We emphasize here that the efficiency of a branch-and-bound algorithm considerably
relies on the quality of the lower bound. To the best of our knowledge, our proposed branch-and-bound algorithm is the first tailored algorithm for globally solving problem (CQP).
\end{itemize}

We apply our proposed branch-and-bound algorithm to solve three important signal processing problems, i.e., the MIMO detection problem  \cite{Jalden,Ma2004}, the unimodular radar code design problem \cite{Maio2009,Maio2011,Soltanalian}, and the virtual beamforming design problem \cite{Hong}. Simulation results show that our proposed new SDR, when applied to these problems, is indeed generally much tighter than the conventional SDR. Moreover, simulation results show that our proposed global algorithm is highly efficient and outperforms the state-of-the-art algorithm/solver for solving these problems. More specifically, our proposed algorithm can solve the MIMO detection problem in case of $8$-PSK and with the number of inputs and outputs $n~\text{and}~m$ being $20$ and signal-to-noise-ratio (SNR) being $5$ dB within {$140$} seconds (on average) while the state-of-the-art sphere decode algorithm \cite{Chan,Damen} needs {$1836$} seconds; our proposed algorithm can solve the virtual beamforming design problem with $m=10$ and $n=5$ within $0.09$ seconds while the state-of-the-art general-purpose global optimization solver Baron \cite{Tawarmalani} needs more than $80$ seconds. %[[[\textbf{Please provide two references here}]]]
%  in the hard case (where the number of inputs and outputs is equal) for solving the MIMO detection problem and  for solving the virtual beamforming design problem.

%: 1) We design an improved semidefinite relaxation that is much tighter than conventional relaxations; 2) We propose an efficient branch-and-bound algorithm for globally solving problem (CQP). The key feature of our methods is to represent the variables in problem (CQP) in their polar coordinate forms, such that the special structure of problem (CQP) can be exploited to derive effective algorithms. Our numerical experiments will demonstrate the effectiveness of the proposed algorithms.

%The rest of this paper is organized as follows. \rev{In Section \ref{sec:application}, we list three signal processing applications of problem (CQP).} In Section \ref{sec:sdr},
%we develop a new and enhanced SDR for problem (CQP). In Section \ref{sec:bbalgorithm}, we propose a tailored branch-and-bound algorithm, based on the enhanced SDR, for solving problem (CQP). Numerical results are presented in Section \ref{sec:numerical} and conclusions are drawn in Section \ref{sec:conclusion}.
%the polar coordinate representations of the complex variables.

%review the conventional SDR for problem (CQP) and

We adopt the following notations throughout the paper. We use lowercase
boldface and uppercase boldface letters to denote (column)
vectors and matrices, respectively. %{\color{red}{(except the special symbol $\textbf{i}$, which denotes  the imaginary unit)}}.
For a given complex vector $\bx\in \mathbb{C}^n$, $\left\|\bx\right\|_2$ denotes its Euclidean norm and
$\mathrm{Re}(\bx)$ and $\mathrm{Im}(\bx)$ denote its component-wise real and imaginary part, respectively.% and  denotes  component-wise  part.
%{For a given complex number $c,$ the notations $\mathrm{R}(c)$ and $\mathrm{I}(c)$ stand for its real and imaginary parts, respectively. %For a given (complex) vector $\bx$, $\|\bx\|$ denotes its Euclidean norm.
~For a given complex Hermitian matrix $\bA$, $\bA\succeq \mathbf{0}$ means $\bA$ is positive semidefinite. For two given Hermitian matrices $\bA$ and $\bB$, $\bA\succeq \bB$ means $\bA-\bB\succeq \mathbf{0}.$ Moreover, let $\textrm{Trace}(\cdot)$ denote the trace operator, let $\bA\bullet \bB$ denote $\textrm{Trace}(\bA^{\dagger}\bB)$ (i.e., $\sum_{i}\sum_{j}A_{ij}B_{ij},$ where $A_{ij}$ denotes the $(i,j)$-th entry of matrix $\bA$), and let $\left\|\bA\right\|_{\mathrm{F}}$ denote $\sqrt{\bA\bullet \bA}.$ Finally, we use $\nu^*$ to denote the optimal value of problem (CQP).
%we use $\mathbf{i}$ to denote the imaginary unit which satisfies the equation $\mathbf{i}^2 = -1.$

%Throughout the paper, we use the following notations: For a vector $\bx\in \mathbb{C}^n$, $\|\bx\|=\sqrt{\sum_{i=1}^n |x_i|^2}$ denotes its Euclidean-norm,
%$\mathrm{Re}(\bx)$ denotes the component-wise real part, and $\mathrm{Im}(\bx)$ denotes the component-wise imaginary part.
%For a complex Hermitian matrix $\bA$, $\bA\succeq 0$ means $\bA$
%is positive semidefinite, and $\textrm{Trace}(\bA)$ denotes the trace of $\bA$. For two Hermitian matrices $\bA$
%and $\bB$, $\bA\succeq \bB$ means $\bA-\bB\succeq 0$. Denote $\bA\cdot \bB=\textrm{Trace}\left(\bA^{\dag}\bB \right)$, and the matrix norm
%is given by $\|\bA\|=\sqrt{\bA\cdot \bA}$.

{\section{Three Signal Processing Applications of Problem (CQP)}\label{sec:application}}
%To show the importance of problem (CQP),
{In this section}, we list three important signal processing applications of problem (CQP) and we will focus on these three applications throughout the paper.

$\blacksquare$ \textbf{MIMO detection} \cite{Jalden,Ma2004}. The input-output relationship of the MIMO channel can be modeled as $\br=\bH\bx+\bv,$
where $\bH\in\mathbb{C}^{m\times n}$ is the complex channel matrix (for $n$ inputs and $m$ outputs with $m\geq n$), $\bv\in\mathbb{C}^{m}$ is the additive white Gaussian noise, $\br \in\mathbb{C}^{m}$ is the vector of received signals, and $\bx\in\mathbb{C}^{n}$ is the vector of transmitted symbols.
Assume the $M$-Phase-Shift Keying (PSK) modulation scheme with $M\geq 2$ is adopted. Then, each entry $x_i$ of $\bx$ belongs to a finite set of symbols, i.e., {$x_i\in \left\{\exp\left({\textsf{i}\theta}\right)\mid\theta\in \mathcal{A}\right\},$ where $\textsf{i}$ is the imaginary unit satisfying $\textsf{i}^2=-1$} and
$$\mathcal{A}=\left\{\theta\mid\theta=2k\pi/M,~k=0,1,\ldots,M-1\right\}.$$
The maximum likelihood MIMO detection problem is %formulated as
\begin{align}
\min_{\bx \in \mathbb{C}^{n}} ~&~ \frac{1}{2}\left\|\bH\bx-\br\right\|_2^2 \label{MIMO-Problem}\\ \nonumber%tag{ML}
\mbox{s.t.~} ~&~ \left|x_i\right|=1,~ \arg\left(x_i\right)\in \mathcal{A},~i=1,2,\ldots,n.
\end{align}
%%
%\begin{equation}
%\begin{array}{rl}
%\displaystyle \min_{\mathbf{x} \in \mathbb{C}^{n}} &\left\|\mathbf{H}\mathbf{x}-\mathbf{r}\right\|_2^2 \\[-2pt]
%\mbox{s.t.~} & x_i\in \mathcal{S}_M,~i=1,2,\ldots,n.
%\end{array}\end{equation}
%
%$\mathcal{S}_M:=\left\{\exp(\text{i}\theta)|\theta=2k\pi/M,k=0,1,\ldots,M-1\right\}$
%
%$\min_{\mathbf{x} \in \mathbb{C}^{n}}\left\|\mathbf{H}\mathbf{x}-\mathbf{r}\right\|_2^2,s.t.x_i\in \mathcal{S}_M,i=1,2,\ldots,n,$
%$\mathbf{z}$ $\mathbf{H}\in\mathbb{C}^{m\times n}$
%
%%%\begin{align}
%\min_{\bx} ~&~ \frac{1}{2}\left\|\bH\bx-\br\right\|_2^2 \label{MIMO-Problem}\\ \nonumber%tag{ML}
%\mbox{s.t.~} ~&~ \left|x_i\right|=1,~ \arg\left(x_i\right)\in \mathcal{A},~i=1,\ldots,n.
%\end{align}
It is clear that problem \eqref{MIMO-Problem} is a special case of problem (CQP) with
$\bQ=\bH^{\dagger} \bH,~\bc=-\bH^\dagger \br,~\ell_i=u_i=1,~\mathcal{A}_i=\mathcal{A},~ i=1,2,\ldots,n.$ In this case, all of $\mathcal{A}_i$ in problem (CQP) are discrete sets.

$\blacksquare$ \textbf{Unimodular radar code design} \cite{Maio2009,Maio2011,Soltanalian}. %[[[\rev{\textbf{Cheng: it would be great if we can give more physical interpretation of this problem as in the other two examples/applications.}}]]]
The goal of the unimodular radar code design problem is to maximize the system's detection performance under the similarity constraint (for controlling the ambiguity distortion).
It has been shown (e.g., in \cite{Maio2009}) that the system's detection performance depends on the radar code, the disturbance
covariance matrix, and the temporal steering vector only through the SNR. Mathematically, the SNR of the considered radar system can be expressed as $$c \cdot\bx^\dag \left(\bM^{-1}\odot \left(\bp\bp^\dag \right)^*\right) \bx,$$ where $c$ is a constant (depending only on the cases of the nonfluctuaing and fluctuating target), $\bx\in\mathbb{C}^n$ is the unimodular radar code to be designed, $\bM$ is the positive definite covariance matrix of some unknown zero-mean complex Gaussian noise vector, $\bp=\left[1,\,e^{\textsf{i}2\pi f_dT_r},\ldots,e^{\textsf{i}2\pi (n-1) f_dT_r}\right]^\T$ is the temporal steering vector with $f_d$ being the target Doppler frequency, $T_r$ being the pulse repetition time, and $N$ being the length of the radar code. In the above, $\odot$ denotes the Hadamard product operator, $(\cdot)^{-1}$ denotes the (matrix) inverse operator, and $(\cdot)^*$ denotes the element-wise conjugate operator. The unimodular radar code design problem can be formulated as %follows (see \cite{Maio2009}):
\begin{equation}\label{Radar-Problem}\begin{array}{cl}
\displaystyle \max_{\bx\in\mathbb{C}^n} & \bx^{\dag}\left(\bM^{-1}\odot \left(\bp\bp^\dag \right)^*\right)\bx \\[3pt]
\mbox{s.t.} &\displaystyle |x_i|=1,~i=1,2,\ldots,n, \\[5pt]
            &\displaystyle \|\bx-\bx^0\|_{\infty}\leq \delta,
\end{array}\end{equation}
where $\bx^0 \in \left\{\bx\in\mathbb{C}^n\mid|x_i|=1,~i=1,2,\ldots,n\right\}$ is a predefined desired radar code (e.g., the Barker code) and $\delta>0$ is a given similarity tolerance. When the tolerance $\delta$ is small enough (e.g., $\delta<\sqrt{2}$), the constraint $\|\bx-\bx^0\|_{\infty}\leq \delta$ is equivalent to
$\arg \left(x_i\right) \in [\underline{\theta}_i,\bar{\theta}_i]$ for all $i=1,2,\ldots,n$, where $$\underline{\theta}_i=\arg \left(x^0_i\right)-\arccos(1-\delta^2/2),$$ $$\bar{\theta}_i=\arg \left(x^0_i\right)+\arccos(1-\delta^2/2).$$ Therefore, the above unimodular radar code design
problem \eqref{Radar-Problem} is also a special case of problem (CQP) with $\bQ=-2\bM^{-1}\odot \left(\bp\bp^\dag \right)^*,~\bc=\bm{0},~l_i=u_i=1,$ and $\mathcal{A}_i = [\underline{\theta}_i,\bar{\theta}_i]$ for $i=1,2,\ldots,n$. In this case, all of $\mathcal{A}_i$ in problem (CQP) are continuous sets.

%$\blacksquare$ \textbf{Unit-modulus constrained quadratic programming \cite{Lu2018}.}
%A special case of problem (CQP) is the following unit-modulus constrained quadratic programming problem:
%\begin{align*}
%\min_{\bx} ~&~ \frac{1}{2} x^{\dag}Qx+\mathrm{Re}(c^{\dag}x) \\ \tag{UQP}
%\mbox{s.t.} ~~ &|x_i|^2=1,~i=1,...,n,\\
%&\arg\left(x_i\right)\in \mathcal{A}_i,~i=1,...,n.
%\end{align*}
%In fact, the MIMO channel detection problem \eqref{MIMO-Problem} and the radar code design problem  \eqref{Radar-Problem} are all special cases of problem (UQP).
%Besides of the above two applications, problem (UQP) has special interest in some phase angle related signal processing problems, including the
%phase recovery problem \cite{Waldspurger} and the phase synchronization problem \cite{Boumal2016}.
%Moreover, the classical max-cut problem \cite{Goemans1995} and max-3-cut problem \cite{Goemans2004}
%can all be formulated as (UQP).

$\blacksquare$ \textbf{Virtual {beamforming} design \cite{Hong}.}
Suppose that there is a set $\{1,2,\ldots,n\}$ of transmitters each equipped with a single antenna and there is a single receiver equipped with $m$ receive antennas. Suppose that all $n$ transmitters can fully cooperate with each other and let $\bx\in\mathbb{C}^n$ be the virtual transmit beamforming vector formed by all transmitters. Let $\bh_j\in \mathbb{C}^n$ be
the channel vector between all transmitters and the $j$-th antenna of the receiver. %, and let $x$ be the beamforming vector.
The virtual beamforming design problem in this single-hop wireless network is to maximize the total received signal power subject to
individual transmit power constraints. %$|x_i|\leq P$.
Mathematically, the problem can be formulated as
\begin{align}
\max_{\bx\in\mathbb{C}^n} ~&~ \sum_{j=1}^m\left|\bh_j^\dagger \bx\right|^2 \label{VB}\\\nonumber%\tag{VB}
\mbox{s.t.~~} ~&~ |x_i|\leq \sqrt{P_i},~i=1,2,\ldots,n,
\end{align}
%\begin{align}
%\max_{\bx} ~&~ \sum_{j=1}^m\left|\bh_j^\dagger \bx\right|^2 \label{VB}\\\nonumber%\tag{VB}
%\mbox{s.t.~} ~&~ |x_i|\leq \sqrt{P_i},~i=1,2,\ldots,n,
%\end{align}
where $P_i$ is the power budget of transmitter $i$.
Again, problem \eqref{VB} is a special case of problem (CQP) with
$\bQ=-2 \sum_{j=1}^m \bh_j \bh_j^\dagger,~\bc=\bm{0},~\ell_i=0,~u_i=\sqrt{P_i},~\mathcal{A}_i=[0,2\pi],~i=1,2,\ldots,n.$

\section{An Enhanced SDR for Problem (CQP)}\label{sec:sdr}
In this section, we first review the conventional SDR and then develop a new enhanced SDR for problem (CQP).

\subsection{Conventional SDR}
By introducing an $n\times n$ complex matrix $\bX=\bx\bx^\dag$, problem (CQP) can be equivalently reformulated as
\begin{align*}
\min_{\bx,\bX} ~&~ \frac{1}{2} \bQ\bullet \bX+\mathrm{Re}\left(\bc^{\dagger}\bx\right) \\ \tag{P}
\mbox{s.t.~} ~&~ \ell^2_i \leq X_{ii}\leq u^2_i,~i=1,2,\ldots,n,\\
&~\arg\left(x_i\right)\in \mathcal{A}_i,~i=1,2,\ldots,n,\\
&~ \bX=\bx\bx^\dagger,
\end{align*}
%\begin{align*}
%\min_{\bx,\bX} ~&~ \frac{1}{2} \bQ\cdot \bX+\mathrm{Re}\left(\bc^{\dag}\bx \right) \\ \tag{P}
%\mbox{s.t.} ~&~ \ell^2_i \leq X_{ii}\leq u^2_i,~i=1,2,\ldots,n,\\
%&~\arg \left(x_i\right)\in \mathcal{A}_i,~i=1,2,\ldots,n,\\
%&~ \bX=\bx\bx^\dag,
%\end{align*}
where $X_{ii}$ is the $i$-th diagonal entry of $\bX.$ The conventional SDR
of problem (P) is
\begin{align*}
\min_{\bx,\bX} ~&~ \frac{1}{2} \bQ\bullet \bX+\mathrm{Re}\left(\bc^{\dagger}\bx\right) \\ \tag{CSDR}
\mbox{s.t.~} ~&~ \ell^2_i \leq X_{ii}\leq u^2_i,~i=1,2,\ldots,n,\\
&~ \bX\succeq \bx\bx^\dagger,
\end{align*}
%\begin{align*}
%\min_{\bx,\bX} ~&~ \frac{1}{2} \bQ\cdot \bX+\mathrm{Re}\left(\bc^{\dag}\bx \right) \\ \tag{CSDR}
%\mbox{s.t.} ~&~ \ell^2_i \leq X_{ii}\leq u^2_i,~i=1,2,\ldots,n,\\
%&~ \bX\succeq \bx\bx^\dag,
%\end{align*}
which  relaxes $\bX=\bx\bx^\dagger$ to $\bX\succeq \bx\bx^\dagger$ and drops the argument constraints
$\arg\left(x_i\right)\in \mathcal{A}_i$ for all $i=1,2,\ldots,n.$

The above conventional relaxation (CSDR) has been widely applied for solving problems arising from signal processing and other applications.
Based on (CSDR), various (approximation) algorithms have been proposed. Indeed, all the (approximation) algorithms proposed in \cite{Ma2004,Maio2009,Maio2011,Waldspurger,Bandeira,Luo2010,So2008,Zhang2006,Gershman,He,Luo2007}
are based on (CSDR) (or its equivalent reformulations). %[[[\textbf{what do we mean by ``equivalent forms" here?}]]]

{{Problem (CSDR) can be solved in polynomial time by using the interior-point algorithm (for any given positive error tolerance) \cite[Sec. 6.6.3]{BenTal}}}, and its optimal value serves as a lower bound of problem (P). If the optimal solution $(\bar{\bx},\bar{\bX})$ of problem (CSDR)
is of rank one, i.e., satisfying $\bar\bX=\bar\bx\bar\bx^\dagger,$ then $(\bar\bx,\bar\bX)$ is the global solution of problem (P) without the argument constraints. However, $\bar\bX$ might not be of rank one and $\bar\bx$ might not satisfy the argument constraints. In these cases, there is a nonzero gap between problem (P) and its relaxation (CSDR){, where the gap between two problems in this paper refers to the absolute value of the difference between the optimal values of the two problems.}

% {\color{red}{(i.e., the absolute value of the difference between the optimal values of the two problems)}}.

\subsection{An Enhanced SDR}

In this subsection, we derive more valid inequalities to reduce the gap between problems (P) and (CSDR) and develop an enhanced SDR for problem (P).

%To identify valid inequalities, we study the difference between the feasible set of (P)
%\begin{equation}\label{eq1}
%\left\{\left(\bx,\bX\right)|~\bX = \bx\bx^\dag, l_i \leq |x_i| \leq u_i,\arg \left(x_i\right)\in \mathcal{A}_i,i=1,2,\ldots,n\right\}
%\end{equation}
%and the feasible set of (CSDR)
%\begin{equation}\label{eq2}
%\left\{(\bx,\bX)|~\bX\succeq \bx\bx^\dag, \ell^2_i \leq X_{ii}\leq u^2_i,~i=1,2,\ldots,n\right\}.
%\end{equation}
Notice that the nonconvex equality constraint $\bX=\bx\bx^\dag$ in problem (P) can be equivalently reformulated as
\begin{equation}\label{eq3}
\bX\succeq \bx \bx^\dag~\textrm{and}~X_{ii}=|x_{i}|^2,~i=1,2,\ldots,n.
\end{equation}
In fact, if $(\bx,\bX)$ satisfies \eqref{eq3}, then $\bX-\bx\bx^\dag$ is a positive semidefinite matrix
with all diagonal entries being zero, and thus $\bX=\bx\bx^\dag$.
Now, we introduce the polar coordinate representation $x_i=r_i e^{\textsf{i}\theta_i}$ of the complex variable $x_i$ for all $i=1,2,\ldots,n.$ By the equivalence of $\bX=\bx\bx^\dag$ and \eqref{eq3}, we get the following proposition.

\begin{proposition}
The feasible set of problem (P) can be equivalently expressed as follows: $\bX\succeq \bx\bx^\dagger$ and
$$\ell_i\leq r_{i}\leq u_i,~X_{ii}=r_i^2,~x_i=r_ie^{\textsf{i}\theta_i},~\theta_i\in\mathcal{A}_i,~i=1,2,\ldots,n.$$
\end{proposition} %[[[\textbf{Any reasons to write $\ell_i\leq r_{i}\leq u_i$? Maybe we can add both $\ell^2_i \leq X_{ii}\leq u^2_i$ and $\ell_i\leq r_{i}\leq u_i$ in the new relaxation, because after the relaxation they generally are not equivalent, right?}]]]
\begin{proof} We first show that, if $(\bx,\bX,\br)$ and $\left\{\theta_i\right\}$ satisfy all conditions in the proposition, then $(\bx,\bX)$ is feasible to problem (P).
 First, the assumption immediately implies that $(\bx,\bX)$ satisfies the first two constraints in problem (P). It remains to show $\bX=\bx\bx^\dag.$ This can be immediately obtained by combining $\bX\succeq \bx\bx^\dagger$ and $X_{ii}=r_i^2=\left|x_{i}\right|^2,~i=1,2,\ldots,n.$ %we immediately get
%imply that $\bx$ is feasible to (CQP), and
%imply that  Then, $(\bx,\bX)$ is feasible to (P).
For the converse direction, assume that $(\bx,\bX)$ is a feasible solution of (P). Let $r_i=|x_i|$ and $\theta_i=\arg \left(x_i\right)$ for all $i=1,2,\ldots,n.$ Then, it is simple to check
$(\bx,\bX,\br)$ jointly with $\left\{\theta_i\right\}$ satisfy all conditions in the proposition. %The proof is complete.
\end{proof}

The constraints $x_i=r_ie^{\textsf{i}\theta_i},~\theta_i\in\mathcal{A}_i$ and  $X_{ii}=r_i^2$ in Proposition 1 are still not convex, but they allow for simple convex relaxations. Below we derive convex envelopes of these two types of nonconvex constraints, which lead to a new tighter SDR for problem (P).

Let us first consider the nonconvex set
\begin{equation}\label{nonconvexsetAi}{\mathcal{S}_{\mathcal{A}_i}:=\left\{(x_i ,r_i)\,|\,x_i =r_i e^{\textsf{i}\theta_i},~\theta_i\in \mathcal{A}_i,~r_i\geq0\right\}}\end{equation}
and let $\mathcal{G}_{\mathcal{A}_i}$ be its convex envelope. By using the similar arguments in \cite{Lu2018,lu2017efficient}, one can show that: (1) if $\mathcal{A}_i=[\underline {\theta}_i, \bar\theta_i]$ with $\bar\theta_i-\underline {\theta}_i \leq \pi,$ then
\begin{equation}\label{GAi1}\small{\mathcal{G}_{\mathcal{A}_i}=\left\{(x_i,r_i)\mid\left|x_i\right|\leq r_i,\,\alpha_{i} \mathrm{Re}\left(x_i\right)+\beta_{i} \mathrm{Im}\left(x_i\right)\geq \gamma_{i}r_i\right\}},\end{equation}
where
%\begin{equation}\label{alpha}
%\small{\alpha_{i}=\cos\left(\frac{\underline {\theta}_i+\bar \theta_i}{2}\right), \beta_{i}=\sin\left(\frac{\underline {\theta}_i+\bar \theta_i}{2}\right),
%\gamma_{i}=\cos\left(\frac{\underline {\theta}_i-\bar \theta_i}{2}\right)};
%\end{equation}
\begin{equation}\label{alpha}
\begin{array}{rl}
&\alpha_{i}=\cos\left(\frac{\underline {\theta}_i+\bar \theta_i}{2}\right), \beta_{i}=\sin\left(\frac{\underline {\theta}_i+\bar \theta_i}{2}\right),\\[5pt]
&\gamma_{i}=\cos\left(\frac{\underline {\theta}_i-\bar \theta_i}{2}\right);
\end{array}
\end{equation}
%(2) if $\mathcal{A}_i=[\underline {\theta}_i,\bar\theta_i]$ with $\bar\theta_i-\underline {\theta}_i = \pi,$ then %$\mathcal{G}_{[\ell_i,u_i]}$ has the form
%\begin{equation}\label{GAi2}\mathcal{G}_{\mathcal{A}_i}=\left\{(x_i,r_i)\mid\left|x_i\right|\leq r_i,\,\alpha_i \mathrm{Re}\left(x_i\right)+\beta_i \mathrm{Im}\left(x_i\right)\geq 0 \right\},\end{equation} where $\alpha_{i}$ and $\beta_{i}$ are given in \eqref{alpha};
%$c_i=\cos\left(\frac{\underline {\theta}_i+\bar\theta_i}{2}\right)~\text{and}~d_i=\sin\left(\frac{\underline {\theta}_i+\bar\theta_i}{2}\right);$
(2) if $\mathcal{A}_i=\{\theta_{i}^1,\theta_{i}^2,\ldots,\theta_{i}^M\}$ is a discrete set (with a finite number of elements) with $0\leq \theta_{i}^1 < \theta_{i}^2 < \cdots<\theta_{i}^M<2\pi$ and $\theta_{i}^{j+1}-\theta_{i}^{j} \leq \pi$ for all $i=1,2,\ldots, M-1,$
then $\mathcal{G}_{\mathcal{A}_i}$ is a polyhedral set and %, whose extreme directions include
%$\left\{\left(e^{\textbf{i}\theta_1},1\right),\left(e^{\textbf{i}\theta_2},1\right),\ldots,\left(e^{\textbf{i}\theta_M},1\right)\right\}$.
%Then the polyhedral set $\mathcal{G}_{\mathcal{A}_i}$ can be represented by
\begin{equation}
\mathcal{G}_{\mathcal{A}_i}=\left\{(x_i,r_i)\,\left|\, \begin{array}{@{}lll} \alpha_i^j \mathrm{Re}\left(x_i\right)+ \beta_i^j \mathrm{Im}\left(x_i\right)\leq \gamma_i^j r_i,\\j=1,2,\ldots,M \end{array}\right.\right\},
\end{equation}
where
\begin{align*}
&\alpha_i^j=\cos\left(\frac{\theta_{i}^j+\theta_i^{j+1}}{2}\right),~\beta_i^j=\sin\left(\frac{\theta_i^{j}+\theta_i^{j+1}}{2}\right),\\[2pt]
&\gamma_i^j=\cos\left(\frac{\theta_i^{j+1}-\theta_i^{j}}{2}\right),
\end{align*} and $\theta_i^{M+1}=\theta_i^{1}+2\pi.$ {{
For any $r\geq 0,$ define $\mathcal{G}_{\mathcal{A}_i}(r)= \left\{ x_i \,|\, (x_i,r_i)\in \mathcal{G}_{\mathcal{A}_i},~r_i=r \right\}.$
Then, it is simple to see that $\mathcal{G}_{\mathcal{A}_i}(r)$ is a slice of $\mathcal{G}_{\mathcal{A}_i}$ with $r_i=r$ and
$\mathcal{G}_{\mathcal{A}_i} =\bigcup_{r\geq 0} \left\{ (x_i,r_i) \,|\, x_i \in \mathcal{G}_{\mathcal{A}_i}(r) \right\}.$ %Thus, by
%understanding the structure of set $\mathcal{X}_{\mathcal{A}_i}(t)$, we can easily imagine the structure of $\mathcal{G}_{\mathcal{A}_i}$.
An illustration of how $\mathcal{G}_{\mathcal{A}_i}(1)$ looks like for both continuous and discrete sets $\mathcal{A}_i$ is given in Fig. 1. %[[[\textbf{Is it possible to change x and y in Fig. 1 into $Re(x_i)$ and $Im(x_i)$?}]]]
}}
%for $j=1,2,\ldots,M.$ For simplicity, let $\theta_{M+1}:=\theta_{1}+2\pi$.
%
%if $\mathcal{A}_i$ is a finite set, then $\mathcal{G}_{\mathcal{A}_i}$ is a polyhedral set. [[[\textbf{Cheng: Please provide more details here. Let us give the general case with $M$ different $\theta_i$ satisfying $0\leq \theta_1 < \theta_2 < \cdots<\theta_M<2\pi$.}]]]

Now, let us consider the nonconvex set
{$$\left\{(X_{ii},r_i)\mid X_{ii}=r_i^2,~r_i\in \mathcal{B}_i\right\},$$}where $\mathcal{B}_i=[\ell_i,u_i]$.
Let $\mathcal{F}_{\mathcal{B}_i}$ be its convex envelope. We can show that
\begin{equation}\label{FBi}\mathcal{F}_{\mathcal{B}_i}=\left\{(X_{ii},r_i) \left| \begin{array}{@{}lll} X_{ii} \geq r_i^2,\\X_{ii} -(\ell_i+u_i)r_i +\ell_i u_i\leq 0\end{array}\right.\right\}.\end{equation} See Fig. 2
for an illustration of $\mathcal{F}_{\mathcal{B}_i}$. %[[[\textbf{Please add $X_{ii}$ and $r_i$ in Fig. 2.}]]]

\begin{figure}[!t]
\centering
\includegraphics[width=7.4cm]{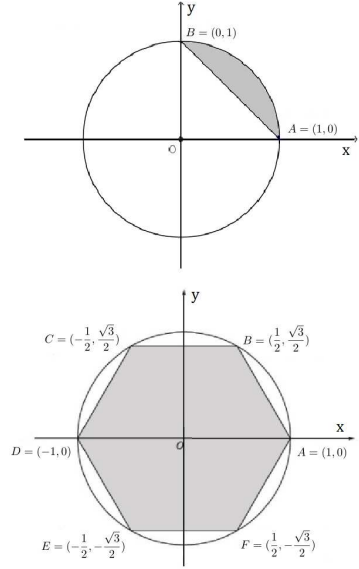}
\caption{{{An illustration of the set $\mathcal{G}_{\mathcal{A}_i}(1)$}}, where the top one corresponds to the continuous case where $\mathcal{A}_i=[0, \pi/2]$ and the bottom one corresponds to the discrete case where $\mathcal{A}_i=\left\{0, \pi/3, 2\pi/3, \pi, 4\pi/3, 5\pi/3\right\}$.}
%Illustration to the set $\left\{x_i|~(x_i,r_i)\in\mathcal{G}_{\mathcal{A}_i}\right\}$ for a fixed $r_i>0$. Upper is for continuous case, Lower
%is for discrete case.}
\end{figure}

\begin{figure}[!t]
\centering
\includegraphics[width=7.4cm]{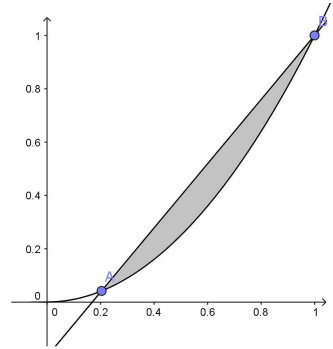}
\caption{An illustration of $\mathcal{F}_{\mathcal{B}_i}$ where $\mathcal{B}_i=[0.2, 1].$}
%Illustration to the set $\mathcal{F}_{\mathcal{B}_i}$.}
\end{figure}

Based on the above convex envelopes of two different types of nonconvex constraints, we can obtain the following enhanced SDR for problem (P):
\begin{align*}
\min_{\bx,\bX,\br} ~&~ \frac{1}{2} \bQ\bullet \bX+\mathrm{Re} \left(\bc^{\dagger}\bx\right) \\
\mbox{s.t.~} ~&~ \ell_i \leq r_{i}\leq u_i,~i=1,2,\ldots,n,\\ \tag{ECSDR}
%~&~ \ell_i^2 \leq X_{ii}\leq u_i^2,~i=1,2,\ldots,n,\\
&~(x_i,r_i)\in\mathcal{G}_{\mathcal{A}_i},~i=1,2,\ldots,n,\\
&~(X_{ii},r_i)\in \mathcal{F}_{\mathcal{B}_i},~i=1,2,\ldots,n,\\
&~ \bX\succeq \bx\bx^\dag,
\end{align*}
where $\br=[r_1,r_2,\ldots,r_n]^\T.$ %and $\bm{\theta}=[\theta_1,...\ldots,\theta_n]^T$.
Note that although $\theta_i~(i=1,2,\ldots,n)$ do not explicitly appear in problem (ECSDR), they play an important role in defining it. This is because the sets $\mathcal{G}_{\mathcal{A}_i}~(i=1,2,\ldots,n)$ and thus problem (ECSDR) are determined by the range $\mathcal{A}_i$ of $\theta_i.$
Throughout the paper, we denote $$\mathcal{D}=\prod_{i=1}^n \mathcal{A}_i \times \prod_{i=1}^n \mathcal{B}_i$$ as the collection of $\mathcal{A}_i$ and $\mathcal{B}_i$ for all $i=1,2,\ldots,n$
and denote ECSDR($\mathcal{D}$) as the corresponding instance of problem (ECSDR) defined over $\mathcal{D}$.

Note that the constraints $(X_{ii},r_i)\in \mathcal{F}_{\mathcal{B}_i}$ and $\ell_i \leq r_{i}\leq u_i$ in (ECSDR) imply $\ell^2_i \leq X_{ii}\leq u^2_i.$ Therefore,
problem (ECSDR) is generally a tighter relaxation for problem (P) than (CSDR)\footnote{The only case under which problems (ECSDR) and (CSDR) are equivalent is $\mathcal{A}_i=[0, 2\pi]$ for all $i=1,2,\ldots,n.$ See our discussion on this special case at the end of this section.}. First, (CSDR) completely neglects the argument constraints in problem (P) whereas the constraints $(x_i,r_i)\in\mathcal{G}_{\mathcal{A}_i}~(i=1,2,\ldots,n)$ in (ECSDR) carefully exploit the argument information. Moreover, the constraint $\bX=\bx\bx^\dagger$ is relaxed to $\bX\succeq \bx\bx^\dagger$ in (CSDR). The gap due to this relaxation is reduced in (ECSDR) because valid constraints $(X_{ii},r_i)\in \mathcal{F}_{\mathcal{B}_i}~(i=1,2,\ldots,n)$ are added in (ECSDR).

\subsection{Tightness and Relaxation Gap of (ECSDR)}
In this subsection, we study the tightness and the relaxation gap of the proposed (ECSDR). We first state the following proposition.

\begin{proposition}
%Assume that $\mathcal{A}_i=[\underline {\theta}_i, \bar\theta_i]$ with $\underline {\theta}_i<\bar\theta_i$, or  $\mathcal{A}_i$ is a discrete set with a finite number of elements.
{For each $i=1,2,\ldots,n,$} if $(x_i,r_i)\in\mathcal{G}_{\mathcal{A}_i}$ and $|x_i|=r_i$, then $\arg \left(x_i\right) \in \mathcal{A}_i$.
\end{proposition}
%\begin{proof}
%{\color{red}{Recall that $\mathcal{G}_{\mathcal{A}_i}$ is the convex envelope of the set $\mathcal{S}_{\mathcal{A}_i}$ in \eqref{nonconvexsetAi}. %$$\left\{(x_i,r_i)\,|\,x_i=r_i e^{\textbf{i}\theta_i},~\theta_i\in \mathcal{A}_i,~r_i\geq0\right\}.$$
%Then, for any $(x_i, r_i)\in \mathcal{G}_{\mathcal{A}_i},$ there must exist $(x_i^j, r_i^j)\in \mathcal{S}_{\mathcal{A}_i}$ and $\alpha_i^j\geq 0, j\in\mathcal{J}$
%%Let $(x,r)\in \mathcal{G}_{\mathcal{A}_i}$. By definition, there exist
%%$$(x_j,r_j)\in \left\{(x_i,r_i)\,|\,x_i=r_i e^{\textbf{i}\theta_i},~\theta_i\in \mathcal{A}_i,~r_i\geq0\right\}$$
%%for $j=1,\ldots,k$,
%such that
%$$x_i=\sum_{j\in\mathcal{J}}\alpha_i^j x_i^j~\text{and}~r_i=\sum_{j\in\mathcal{J}}\alpha_i^j r_i^j,$$ where
%$\sum_{j\in\mathcal{J}}\alpha_i^j=1.$ One can check that
%$$\left|x_i\right| = \left|\sum_{j\in\mathcal{J}}\alpha_i^j x_i^j\right| \leq \sum_{j\in\mathcal{J}}\alpha_i^j \left|x_i^j\right| = \sum_{j\in\mathcal{J}}\alpha_i^j r_i^j = r_i$$
%and the above inequality holds with equality if and only if all of $x_i^j$ have the same argument, which is also the argument of $x_i.$
%Therefore, if $|x_i|=r_i$, we must have $\arg \left(x_i\right) \in \mathcal{A}_i$}.}
%%
%%
%%$\mu_1,\ldots,\mu_k> 0$ and $\mu_1+\ldots+\mu_k=1$.
%%Note that $$|x|=|\mu_1 x_1+\cdots+\mu_k x_k | \leq \mu_1 r_1+\cdots+\mu_k r_k =r,$$
%%the equality holds only when $x_1,\ldots,x_k$ has the same phase angle (which also equals the phase angle of $x$). Hence, if $|x_i|=r_i$, we have $\arg \left(x\right) \in \mathcal{A}_i$}}
%\end{proof}

Instead of providing a rigorous proof for Proposition 2, {{we give an illustration of Proposition 2 using Fig. 1}}.
%For a given $r_i=1\in[\ell_i,u_i]$ with $r_i>0$,
{Consider the set $\mathcal{G}_{\mathcal{A}_i}(1)$ in Fig. 1. %(where $r_i=1$).
It is simple to see from Fig. 1 that if $|x_i|=1$, then $\arg\left(x_i\right)\in \mathcal{A}_i.$} %is satisfied.

The gap between relaxation (ECSDR) and problem (CQP) is generally nonzero. %, so that ECSDR is not exact.
The following theorem presents a tightness result of relaxation (ECSDR).
\begin{theorem}
Let $\left(\bar{\bx},\bar{\bX},\bar{\br}\right)$ be an optimal solution of problem (ECSDR).
If $|\bar{x}_i|=\bar{r}_i$ and $\bar{X}_{ii}=\bar{r}^2_i$ for all $i=1,2,\ldots,n$, then $\bar{\bx}$ is a global solution of problem (CQP) and thus relaxation (ECSDR) is tight.
\end{theorem}
\begin{proof}
Let $\theta_i=\arg \left(\bar{x}_i\right),~i=1,2,\ldots,n$. For each $i=1,2,\ldots,n,$ since $|\bar{x}_i|=\bar{r}_i,$ it follows from Proposition 2 that $\theta_i\in \mathcal{A}_i$
and $\bar x_i=\bar r_i e^{\textsf{i}\theta_i}$. Furthermore, by the assumption that $\bar{X}_{ii}=\bar{r}^2_i$ for all $i=1,2,\ldots,n$ and Proposition 1, we have that $\left(\bar{\bx},\bar{\bX}\right)$ is feasible to problem (P) (and in particular $\bar{\bX}=\bar{\bx}\bar{\bx}^\dag$). Therefore, $\bar{\bx}$ is a global solution of problem (CQP)
and relaxation (ECSDR) is tight.
%Hence,  is a rank-one solution, which implies that
\end{proof}

In the general case, (ECSDR) might not be tight for problem (CQP). In the case that the relaxation
gap is nonzero, it follows from Theorem 1 that there must exist some index $i\in\{1,2,\ldots,n\}$, such that $|\bar{x}_i|<\bar{r}_i$ and/or
$\bar{X}_{ii}>\bar{r}^2_i$.
Let \begin{equation}\label{barundertheta}\bar{\theta}_i=\max \left\{\mathcal{A}_{i}\right\}~\text{and}~\underline{\theta}_i:=\min\left\{\mathcal{A}_{i}\right\}.\end{equation} Next, we provide two tightness estimates of (ECSDR) in the following Proposition~\ref{lemma-width}, whose proof can be found in Appendix \ref{applemmawidth}.

\begin{proposition}\label{lemma-width}
For a given set {$\mathcal{A}_{i} \subseteq [\underline {\theta}_i,\bar\theta_i]$} with $\bar\theta_i-\underline {\theta}_i\leq \pi$,
if $(x_i,r_i)\in\mathcal{G}_{\mathcal{A}_i}$, then \begin{equation}\label{eq10}r_i\geq |x_i|\geq r_i\cos\left(\frac{\bar\theta_i - \underline {\theta}_i}{2}\right).\end{equation}
for a given set $\mathcal{B}_{i}=[\ell_i,u_i]$,
if $(X_{ii},r_i)\in \mathcal{F}_{\mathcal{B}_i}$, then
\begin{equation}\label{eq11}0\leq X_{ii}-r_i^2\leq \frac{(u_i-\ell_i)^2}{4}.\end{equation}
%For a set $\mathcal{A}_{i} \subseteq [\underline{\theta}_i,\bar{\theta}_i]$ with $\bar{\theta}_i-\underline {\theta}_i\leq\pi$,
%if $(x_i,r_i)\in\mathcal{G}_{\mathcal{A}_i}$, then
%\begin{equation}\label{eq10}
%r_i\geq |x_i|\geq r_i\cos\frac{\bar{\theta}_i - \underline{\theta}_i}{2}.
%\end{equation}
\end{proposition}
%\begin{proof}
%See Subsection A of the Appendix.
%\end{proof}

%\begin{lemma}
%For a set $\mathcal{B}_{i}=[\ell_i,u_i]$,
%if $(X_{ii},r_i)\in \mathcal{F}_{\mathcal{B}_i}$, then
%\begin{equation}\label{eq11}
%0\leq X_{ii}-r_i^2\leq \frac{\left(u_i-\ell_i\right)^2}{4}.
%\end{equation}
%\end{lemma}
%\begin{proof}
%See Subsection B of the Appendix.
%\end{proof}

Define the width of $\mathcal{A}_i$ as \begin{equation}\label{WAi}\omega(\mathcal{A}_{i}):=\bar{\theta}_i-\underline{\theta}_i,\end{equation} where $\bar \theta_i$ and $\underline{\theta}_i$ are defined in \eqref{barundertheta},
and the width of $\mathcal{B}_{i}$ as \begin{equation}\label{WBi}\omega(\mathcal{B}_{i}):=u_i-\ell_i.\end{equation}
We can see from Proposition~\ref{lemma-width} that: (1) when $\omega(\mathcal{A}_{i})$ becomes close to zero,
$|x_i|$ will be close to $r_i$ and the constraint $(x_i,r_i)\in\mathcal{G}_{\mathcal{A}_i}$ will be very effective in reducing the difference between $r_i$ and $|x_i|$; (2) similarly, when $\omega(\mathcal{B}_{i})$ becomes close to zero, $X_{ii}$ will be close to $r^2_i$ and
the constraint $(X_{ii},r_i)\in \mathcal{F}_{\mathcal{B}_i}$ will be very effective in reducing the difference between $X_{ii}$ and $r_i^2.$

Now let us discuss two very special cases of (ECSDR). The first case is $\ell_i=u_i=1$ for all $i=1,2,\ldots,n.$ In this case, the set $\mathcal{F}_{\mathcal{B}_i}$ reduces $\mathcal{F}_{\mathcal{B}_i}= \left\{ (X_{ii},r_i)\,|\,X_{ii} = 1,r_i=1\right\}$ and thus the constraint $X_{ii}=r^2_i$ is always satisfied. Hence, if the gap
between problem (CQP) and its relaxation (ECSDR) is nonzero, then the gap must be due to the convex relaxation $(x_i,r_i)\in\mathcal{G}_{\mathcal{A}_i}$,
i.e., there must exist some $i\in\{1,2,\ldots,n\}$ such that $r_i>|x_i|$.
This special case has been studied in \cite{Lu2018,lu2019tightness}. This paper studies a more general problem (CQP) (with interval modulus constraints) and can be regarded as a nontrivial extension from the unit-modular case in \cite{Lu2018,lu2019tightness}. Note that it is straightforward to obtain an SDR for the unit-modulus problem while it is not easy to obtain the SDR for problem (CQP), as the latter requires the use of the polar coordinate representation of the complex variables and Proposition 1.
%
%In particular, the relaxation (ECSDR) is an extension of the SDRs in \cite{Lu2018,lu2019tightness} and the proposed algorithm (to be presented in the next section) is an extension of the algorithms in \cite{Lu2018}.}}
The second case is $\mathcal{A}_{i} =[0,2\pi]$ for all $i=1,2,\ldots,n.$ In this case, we can show that (ECSDR) is equivalent to (CSDR).
To be more specific, the constraints $(x_i,r_i)\in\mathcal{G}_{\mathcal{A}_i}$ and $(X_{ii},r_i)\in \mathcal{F}_{\mathcal{B}_i}$ in this case become $$r_i\geq |x_i|,\,X_{ii} \geq r^2_i,\,X_{ii} -(\ell_i+u_i)r_i+\ell_i u_i\leq 0.$$
Then, for any feasible solution $(\bx,\bX)$ of (CSDR), we can set $r_i=\sqrt{X_{ii}}$  for $i=1,2,\ldots,n$ and check that $(\bx,\bX,\br)$ is a feasible solution of (ECSDR). Therefore, the two relaxations (ECSDR) and (CSDR) are equivalent in this case\footnote{Although relaxation (ECSDR) is not tighter than (CSDR) in this special case,
(ECSDR) still plays an important role of generating the lower bounds in the ECSDR-BB
algorithm in the next section, where as the set $\mathcal{A}_i$ is recursively partitioned into smaller subsets, the quality of relaxation (ECSDR) defined over the subsets will become better than that of (CSDR).}. Except this special case, (ECSDR) is tighter than (CSDR) as discussed before and as will be illustrated later in Section \ref{sec:numerical}.  %[[[\textbf{Maybe we can delete the last two sentences. }]]]

\section{Proposed Global Branch-and-Bound Algorithm}\label{sec:bbalgorithm}

In this section, we propose a global branch-and-bound algorithm based on the enhanced relaxation (ECSDR)
for solving problem (CQP) (equivalent to problem (P)). A typical branch-and-bound algorithm (for a minimization
problem) is generally based on an enumeration procedure, which
partitions the feasible region to smaller subregions and constructs
sub-problems over the partitioned subregions recursively.
In the enumeration procedure, a lower bound for each subproblem
is estimated by solving a relaxation problem. Meanwhile, an
upper bound is obtained from the best known feasible solution
generated by the enumeration procedure or by some other local
optimization/heuristic algorithms. The procedure terminates until the difference between the upper bound and the lower bound is smaller than the given error tolerance $\epsilon>0$, and then an $\epsilon$-optimal solution (defined as below) can be obtained.

\begin{definition}[$\epsilon$-Optimal Solution]\label{esolution}
  Given any $\epsilon>0,$ a feasible point $\bx$ is called an $\epsilon$-optimal solution of problem (CQP) if it satisfies ${F(\bx)-\nu^*} \leq \epsilon.$
\end{definition}

In the remaining part of this section, we first present our proposed branch-and-bound algorithm for solving problem (CQP) in Section III-A.
Then, we show that our proposed branch-and-bound algorithm indeed can find an $\epsilon$-optimal solution of problem (CQP) (for any given $\epsilon>0$) and
analyze its worst-case iteration complexity in Section III-B.

\subsection{Proposed Algorithm}

%which is defined by the following scaling operation:
%\begin{equation}\label{scale}
%\hat{\bx}^k=\textrm{Scale}(\bx^k,\br^k)=\left[r^k_1 e^{\textbf{i}\hat{\theta}^k_1},\ldots,r^k_n e^{\textbf{i}\hat{\theta}^k_n}\right]^{\T},
%\end{equation}
%where $\hat{\theta}^k_i \in \mathcal{A}^k_i$ is the point in $\mathcal{A}^k_i$ that is nearest to $\arg \left(x^k_i\right)$.
%Let $U^*$ denote the upper bound, and $\bx^*$  denote the best known feasible solution of problem (CQP).

To develop a branch-and-bound algorithm for solving problem (CQP), let us first recall Theorem 1. Theorem 1 shows that, if the gap between problem (CQP) and its corresponding relaxation (ECSDR) is not zero, then there must exist some $i\in\{1,2,\ldots,n\}$ with $|\bar{x}_i|<\bar{r}_i$ and/or
$\bar{X}_{ii}>\bar{r}^2_i$. Moreover, Proposition \ref{lemma-width} further shows that we can partition
the sets $\mathcal{A}_i$ and $\mathcal{B}_i$ to reduce the difference $\bar{r}_i-|\bar{x}_i|$ and $\bar{X}_{ii}-\bar{r}^2_i$, respectively.
%, which somehow measures the tightness of the relaxation (ECSDR).
Based on the above observations, we are now ready to present the main steps of the branch-and-bound algorithm. For ease of presentation, we introduce the following notations. Let $\mathcal{D}^0= \prod_{i=1}^n \mathcal{A}_i^0 \times  \prod_{i=1}^n \mathcal{B}_i^0$ with $\mathcal{A}_i^0=\mathcal{A}_i$ and $\mathcal{B}_i=[\ell_i,u_i]$
be the initial feasible set of the polar coordinate variables $\left\{\theta_i\right\}$ and $\br$, and $\mathcal{D}^k=\prod_{i=1}^n \mathcal{A}^k_i \times \prod_{i=1}^n \mathcal{B}^k_i \subseteq \mathcal{D}^0$
be a partitioned subset indexed by $k$.

\textbf{Lower Bound.}
In the branch-and-bound algorithm, the initial feasible set $\mathcal{D}^0$ will be recursively partitioned into smaller subsets.
Obviously, the optimal value $L^k$ of the relaxation problem
ECSDR$\left({\mathcal{D}^k}\right)$ is a lower bound of the optimal value of problem (CQP) defined over the subset $\mathcal{D}^k.$ Therefore, the smallest lower bound among all bounds is a lower bound of the optimal value of the original problem (CQP). This statement will be formally summarized in Theorem \ref{thm-epssolution}.% and proved in Section III-B.

\textbf{Upper Bound.}
%An upper bounding scheme is an important component of the branch-and-bound algorithm is .
An upper bound of problem (CQP) can be obtained by appropriately scaling the solution of any relaxation problem. More specifically, we solve {a} relaxation (ECSDR) (defined over a partitioned subset) %for each node
to obtain its optimal solution $(\bar{\bx},\bar{\bX},\bar{\br})$. Then, we generate a feasible solution of problem (CQP) by using the following scaling operation
%\begin{equation}\label{scale}
%$\hat{\bx}=\textrm{Scale}(\bar{\bx},\bar{\br}),$
%where
\begin{equation}\label{scale}
\hat{\bx}=\mathrm{Scale}(\bar{\bx},\bar{\br}):=\left[\bar{r}_1 e^{\textsf{i}\hat{\theta}_1},\ldots,\bar{r}_n e^{\textsf{i}\hat{\theta}_n}\right]^{\T},
\end{equation}
{{where $$\hat{\theta}_i\in \arg\min_{\theta\in \mathcal{A}_i} \min\left\{|\theta_i- \arg(\bar{x}_i)|,~2\pi-|\theta_i- \arg(\bar{x}_i)|\right\}$$ and $\mathcal{A}_i$ is normalized to satisfy $\mathcal{A}_i\subseteq (-\pi, \pi].$ If the optimal solution to the above problem is not unique, then we just pick one solution.}} %We assume that the problem can be easily solved,
%and if there exist more than one optimal solutions, then $\hat{\theta}_i$ can be any one that randomly selected from these optimal solutions.}}
It is simple to check that the above $\hat{\bx}$ is feasible to problem (CQP). Consequently, $F(\hat{\bx})$ is an upper bound of problem (CQP).
In our branch-and-bound algorithm, we use $U^*$ to denote the best upper bound during the enumeration procedure (i.e., the smallest objective values at all of known feasible solutions at the current iteration) and use $x^*$ to denote the solution that achieves the smallest upper bound. %Once a new subproblem is generated, we
%will check whether a new smaller upper bound is found. If so, then we update the values of $U^*$ and $x^*$ accordingly.

%In the enumeration procedure, different subproblems (defined on different partitioned subsets)
%will be generated. For each generated subproblem, we solve its corresponding relaxation (ECSDR),
%and scale its optimal solution to get a feasible solution of problem (CQP). Since different  subproblems generate different solutions and lead to different upper bounds,

In the proposed branch-and-bound algorithm, we construct a so-called node for each subproblem. The node is denoted as $\left\{\mathcal{D}^k,\bx^{k},\bX^k,\br^k,\hat{\bx}^k,L^{k}\right\}$, in which $\mathcal{D}^k$ is the partitioned subset of the subproblem, $(\bx^{k},\bX^k,\br^k)$ and $L^k$ are
the optimal solution and the optimal value of problem ECSDR$\left({\mathcal{D}^k}\right)$, and $\hat{\bx}^k=\mathrm{Scale}({\bx}^k,{\br}^k)$ (with the operator $\mathrm{Scale}(\cdot,\cdot)$ being defined in \eqref{scale}).

\textbf{Termination Criterion.} If
\begin{equation}\label{rgap}{U^*}-{L^{k}}\leq \epsilon\end{equation}
{{at iteration $k$,}} where $\epsilon$ is the preselected error tolerance, we terminate the algorithm; otherwise we select a node and branch the feasible set of a variable according to some rule. We can see from \eqref{rgap} that, %the efficiency of our proposed algorithm significantly relies on both the lower and upper bounds and
both lower and upper bounds are important to avoid unnecessary branches and enumerations and good lower and upper bounds can significantly improve the computational efficiency of our proposed algorithm. Below, we shall introduce our node selection and branching rules one by one. %as well as lower and upper bounds one by one.% below.

\textbf{Node Selection Rule.} For node $k$, if its lower bound $L^k$
is larger than the upper bound $U^*$, then the global solution of the original problem cannot be located in the set associated with this node. We call a node as an \emph{active} node if its lower bound is smaller than the best known upper bound. Therefore, all of the inactive nodes will not be enumerated in the branch-and-bound algorithm. Let us use $\mathcal{P}$ to denote the set of all active nodes. Our selection rule is to select the active node with the smallest lower bound from $\mathcal{P}$ (to be branched) at each iteration.
%and
%, we. If there
%are more than one active nodes in $\mathcal{P}$, we select the node that has the lowest lower
%bound.

\textbf{Branching Rule.} Let $\left\{\mathcal{D}^k,\bx^{k},\bX^k,\br^k,\hat{\bx}^k,L^{k}\right\}$ be the selected node that has the smallest lower bound in $\mathcal{P}$ and let
\begin{equation}\label{eq18}\begin{array}{l}
i_1^* = \displaystyle \arg\max_{i}\left\{ \left|\hat{x}^k_i-x^k_i\right|\right\},\\[5pt] S^*_1= \displaystyle  \max_{i}\left\{\left|\hat{x}^k_i-x^k_i\right|\right\}, \\[5pt]
i_2^* = \displaystyle \arg\max_{i}\left\{X^k_{ii}-\left(r^k_i \right)^2\right\},\\[5pt] S^*_2=\displaystyle  \max_{i}\left\{X^k_{ii}-\left(r^k_i\right)^2\right\}.
\end{array}
\end{equation} The quantity $\max\left\{S_1^*, S^*_2\right\}$ somehow measures the gap of the corresponding relaxation (ECSDR).
If $S^*_1 \geq S^*_2$, then we select $\mathcal{A}^k_{i_1^*}$ to branch; othewise we select $\mathcal{B}^k_{i_2^*}$ to branch. The selected set is branched into two subsets by
the following rule: if the selected set is an interval, then we partition it into two sub-intervals with equal
lengths; if the selected set is a finite set (in the case where $S^*_1 \geq S^*_2$ and $\mathcal{A}^k_{i_1^*}$ is a
finite set), then we partition the set into two subsets $$\left\{ \theta\mid\theta\in \mathcal{A}^k_{i_1^*}, \theta\leq \theta_{i_1^*}^k \right\}~\text{and}~\left\{ \theta\mid\theta\in \mathcal{A}^k_{i_1^*}, \theta> \theta_{i_1^*}^k \right\},$$ where $$\theta_{i_1^*}^k=\frac{1}{2}\left(\min\left\{ \mathcal{A}^k_{i_1^*}\right\}+\max \left\{\mathcal{A}^k_{i_1^*}\right\}\right).$$
Based on the above rules, we branch the set $\mathcal{D}^k$ into two new sets (denoted as $\mathcal{D}^k_{-}$ and $\mathcal{D}^k_{+}$). It follows from Proposition \ref{lemma-width} that the corresponding relaxation problems defined over the newly obtained two sets $\mathcal{D}^k_{-}$ and $\mathcal{D}^k_{+}$, i.e., the two children problems, are tighter than the one defined over the original set $\mathcal{D}^k.$ Once $\mathcal{D}^k$ has been branched into two sets, the problem instance defined over it will be deleted from the problem list $\mathcal{P}$ and the two children problems will be added into $\mathcal{P}$ if their {{lower bounds}} are less than or equal to the current best upper bound.

%Based on the above procedures, a complete branch-and-bound algorithm can be designed.
By judiciously combining the above main steps, we can obtain
our proposed branch-and-bound algorithm for solving problem (CQP) (equivalent to problem (P)).
The pseudo-codes of our proposed algorithm are given in Algorithm 1. We will call the
algorithm ECSDR-BB (\textbf{E}nhanced \textbf{C}omplex \textbf{S}emi\textbf{D}efinite \textbf{R}elaxation based \textbf{B}ranch-and-\textbf{B}ound) for short.

\begin{figure}
\small{
\textbf{Algorithm 1: ECSDR-BB Algorithm for Solving Problem (CQP)}
\begin{algorithmic}[1]
%\Require
%An instance of CQPP, a given error tolerance $\epsilon>0$, and initial argument constraints $\mathcal{D}^0=\prod_{i=1}^n \mathcal{A}^0_i \times \prod_{i=1}^n \mathcal{B}^0_i:=\prod_{i=1}^n [l_i,u_i] \times \prod_{i=1}^n [0,2\pi]$.
\STATE \textbf{input:} An instance of problem (CQP) and an error tolerance $\epsilon>0.$%, and initial feasible domain $\mathcal{D}^0=\prod_{i=1}^n \mathcal{A}^0_i \times \prod_{i=1}^n \mathcal{B}^0_i:= \prod_{i=1}^n \mathcal{A}_i \times  \prod_{i=1}^n [\ell_i,u_i] $.
\STATE Initialize $\mathcal{P}=\emptyset,$ $\mathcal{D}^0=\prod_{i=1}^n \mathcal{A}^0_i \times \prod_{i=1}^n \mathcal{B}^0_i:=\prod_{i=1}^n \mathcal{A}_i \times \prod_{i=1}^n [\ell_i,u_i],$ and set $k=0.$ // \texttt{Initialization.}
\STATE Solve ECSDR$\left({\mathcal{D}^0}\right)$ for its optimal solution $\left(\bx^{0},\bX^{0},\br^0\right)$ and its optimal value $L^{0}$. // \texttt{Solve relaxation (ECSDR) at the root node.}
\STATE Compute the {feasible point} $\hat{\bx}^0=\mathrm{Scale}(\bx^0,\br^0).$
\STATE Set $U^{*}=F(\hat{\bx}^0)$ and $\bx^{*}=\hat{\bx}^0$. // \texttt{Initial Upper Bound and Optimal Solution.} %, $k=0$. Construct an empty set $\mathcal{P}$.
\STATE Add $\left\{\mathcal{D}^0,\bx^{0},\bX^0,\br^0,\hat{\bx}^0,L^{0}\right\}$ into the node list $\mathcal{P}$.
\LOOP
\STATE Set $k\leftarrow k+1$.
\STATE Using {the \textbf{Node Selection Rule}} to choose a problem from $\mathcal{P},$ denoted as $\{\mathcal{D}^k,\bx^{k},\bX^k,\br^k,\hat{\bx}^k,L^{k}\},$ such that $L^{k}$
 is the smallest one in $\mathcal{P}$. \label{line:lowerbound} // \texttt{Lower Bound.}
\STATE Delete the chosen node from $\mathcal{P}$.
\IF{$U^{*}-L^{k}\leq\epsilon$} %\texttt{{Check the termination criterion \eqref{rgap}}.}
\STATE return $\bx^{*}$ and $U^*$ and terminate the algorithm. \label{line:terminate}// \texttt{{Terminate if \eqref{rgap} is satisfied.}}
\ENDIF
\STATE Choose the set according to \eqref{eq18} and branch $\mathcal{D}^k$ into two subsets $\mathcal{D}^k_{-}$ and $\mathcal{D}^k_{+}$ by using {the \textbf{Branching Rule}}. // \texttt{Branch.}
\FOR{{$s \in \left\{-,+\right\}$}  }
\STATE Solve ECSDR$\left({\mathcal{D}^k_s}\right)$ for its solution $\left(\bx^{k}_{s},\bX^{k}_{s},\br^k_{s}\right)$ and its optimal value $L^{k}_{s}$. // \texttt{{Solve the children problems.}}
\STATE Compute {the feasible point} $\hat{\bx}^k_{s}=\mathrm{Scale}\left(\bx^{k}_{s},\br^k_{s}\right).$
\IF{$U^{*}>F\left(\hat{\bx}^{k}_{s}\right)$}
\STATE set $U^{*}=F\left(\hat{\bx}^{k}_{s}\right)$, $\bx^{*}=\hat{\bx}^{k}_{s}$. \label{line:upper1}// \texttt{Update Upper Bound and Optimal Solution.}
\ENDIF
\IF{$L^{k}_{s}< U^{*}$}
\STATE add $\left\{\mathcal{D}^k_{s},\bx^{k}_{s},\bX^k_{s},\br^k_{s},\hat{\bx}^k_{s},L^{k}_{s}\right\}$ into $\mathcal{P}$.
\ENDIF
\ENDFOR
%\STATE Solve ECSDR$\left({\mathcal{D}^k_+}\right)$ for its solution $\left(\bx^{k}_{+},\bX^{k}_{+},\br^k_{+}\right)$ and its optimal value $L^{k}_{+}$.
%\STATE Compute $\hat{\bx}^k_{+}=\mathrm{Scale}\left(\bx^k_{+},\br^k_{+}\right),$ where the operator $\mathrm{Scale}(\cdot,\cdot)$ is defined in \eqref{scale}.
%\IF{$U^{*}>F(\hat{\bx}^{k}_{+})$}
%\STATE set $U^{*}=F\left(\hat{\bx}^{k}_{+}\right)$, $\bx^{*}=\hat{\bx}^{k}_{+}$. \label{line:upper2}// \texttt{Update Upper Bound and Optimal Solution.}
%\ENDIF
%\IF{$L^{k}_{+}< U^{*}$}
%\STATE add $\left\{\mathcal{D}^k_{+},\bx^{k}_{+},\bX^k_{+},\br^k_{+},\hat{\bx}^k_{+},L^{k}_{+}\right\}$ into $\mathcal{P}$.
%\ENDIF
\ENDLOOP
\end{algorithmic}}
\label{alg_wp}
\end{figure}

%[[[\rev{In the algorithm we state it like $D_$}]]]
\subsection{Global Convergence and Worst-Case Iteration Complexity}

In this subsection, we present some theoretical results of our proposed ECSDR-BB algorithm.

The following Theorem \ref{thm-epssolution} shows that the sequence $\left\{L^k\right\}$ generated by the ECSDR-BB algorithm
is a lower bound of the optimal value of problem (CQP) and the solution $\bx^*$ returned by the algorithm is an $\epsilon$-optimal solution of the problem.
The proof of the theorem  can be found in Appendix \ref{appthmepssolution}. %(CQP).

%We show the proposed ECSDR-BB algorithm must terminate in a finite number of iterations. First, we prove the next result.

\begin{theorem}\label{thm-epssolution}
Let $\left\{\mathcal{D}^k,\bx^{k},\bX^k,\br^k,\hat{\bx}^k,L^{k}\right\}$ be the node selected in Line \ref{line:lowerbound} of the ECSDR-BB algorithm.
Then we have \begin{equation}\label{lowerupper}L^k\leq \nu^*\leq F\left(\hat{\bx}^k\right).\end{equation} Moreover, if \eqref{rgap} holds true, %$F\left(\hat{\bx}^k\right)- L^k\leq \epsilon$,
then $\bx^*$ returned by the algorithm is an $\epsilon$-optimal solution of problem (CQP).
\end{theorem}
%\begin{proof}
%See Subsection C of the Appendix.
%\end{proof}

Next, we will estimate $F(\hat{\bx}^k)- L^k$ and show that \eqref{rgap} will be satisfied after a finite number of iterations.
%Based on Theorem 2, we know that if $F(\hat{\bx}^k)- L^k\leq \epsilon,$ then ECSDR-BB algorithm terminates.
%Now we estimate the difference $F(\hat{\bx}^k)- L^k$.
Define
\begin{equation}\label{umax}
u_{\mathrm{max}}=\max\{u_1,u_2,\ldots,u_n\}
\end{equation}
%and let $i_1^*, i_2^*, S_1^*, S_2^*$ be defined as in \eqref{eq18}.
and the bounded set
\begin{equation}
\mathcal{X}=\left\{\bx\mid|x_i|\leq u_{\max},~i=1,2,\ldots,n\right\}.
\end{equation}
Since $F(\bx)$ is uniformly continuous over the bounded set $\mathcal{X}$, there must exist a constant $M_F>0$ such that
\begin{equation}\label{eq21}
\left|F(\bx)-F(\bx')\right|\leq M_F \left\|\bx-\bx'\right\|_2, \forall~\bx,\bx'\in \mathcal{X}.
\end{equation}
The next lemma gives an upper bound on $F(\hat{\bx}^k)- L^k$, whose proof is relegated to Appendix \ref{applemmabound}.

\begin{lemma}\label{lemmabound}
Let
\begin{equation}\label{eq27}
M_1= \sqrt{n}M_F+   n^{\frac{3}{2}}u_{\max} \left\|\bQ\right\|_{\mathrm{F}}
\end{equation}
and
\begin{equation}\label{eq27c}
M_2= \frac{1}{2} n^{\frac{3}{2}} \left\|\bQ\right\|_{\mathrm{F}},
\end{equation}
where $M_F$ is given in \eqref{eq21}.
%If  $\bar{\theta}^k_{i_1^*}-\underline{\theta}^k_{i_1^*}\leq \pi$, then
Then, we have
\begin{equation}\label{eq20}
F(\hat{\bx}^k)-L^k\leq  M_1 S_1^*+M_2 S_2^*,
\end{equation} where $S_1^*$ and $S_2^*$ are defined in \eqref{eq18}.
%where $\bar{\theta}^k_{i_1^*}=\max \mathcal{A}^k_{i_1^*}$ and $\underline{\theta}^k_{i_1^*}=\min \mathcal{A}^k_{i_1^*}$.
\end{lemma}
%\begin{proof}
%See Subsection D of the Appendix.
%\end{proof}

Based on Lemma \ref{lemmabound}, we can show the following result and we relegate its proof to Appendix \ref{applemmaterminate}.

\begin{lemma}\label{lemmaterminate}
Let
\begin{equation}\label{eq32}
\kappa_1=\left[\frac{8  \epsilon}{u_{\max}(M_1+M_2)}\right]^{\frac{1}{2}}
\end{equation}
and
\begin{equation}\label{eq35}
\kappa_2=\left[\frac{4\epsilon}{M_1+M_2}\right]^{\frac{1}{2}},
\end{equation}
where $M_1$ and $M_2$ are defined in \eqref{eq27} and \eqref{eq27c}, respectively.
If one of the following three conditions is satisfied:\\
(\textrm{C1}) $S_1^*\geq S_2^*$, $\mathcal{A}^k_{i_1^*}=\left[\underline{\theta}^k_{i_1^*},\bar{\theta}^k_{i_1^*}\right]$, and $\bar{\theta}^k_{i_1^*}-\underline{\theta}^k_{i_1^*}\leq \min\{\kappa_1,\pi\}$,\\[2pt]
(\textrm{C2}) $S_1^*\geq S_2^*$, and $\mathcal{A}^k_{i_1^*}$ is a singleton,\\[2pt]
(\textrm{C3}) $S_1^*< S_2^*$, $\mathcal{B}^k_{i_2^*}=\left[\ell^k_{i_2^*},u^k_{i_2^*}\right]$, and $u^k_{i_2^*}-\ell^k_{i_2^*}\leq \kappa_2$,\\[2pt]
then \eqref{rgap} is satisfied and thus the ECSDR-BB algorithm terminates in Line \ref{line:terminate}.
\end{lemma}
%\begin{proof}
%See Subsection E of the Appendix.
%\end{proof}

Based on the above two lemmas, we obtain the main result of this subsection, which shows that the ECSDR-BB algorithm will terminate within a finite number of iterations in \eqref{eq36}.
%prove the next theorem to show the global convergence of ECSDR-BB algorithm.

\begin{theorem}\label{thm-iteration}
For any given error tolerance $\epsilon >0$ and any given instance of problem (CQP), the ECSDR-BB algorithm will return an $\epsilon$-optimal solution of the given instance within at most
\begin{equation}\label{eq36}
\displaystyle {K}:=\prod_{i=1}^n  \left[\mu(\mathcal{A}_i) \times  \max\left\{\left\lceil \frac{2\omega(\mathcal{B}_i)}{\kappa_2}\right\rceil,1\right\}\right]
\end{equation}
iterations, where
\begin{equation}\label{muA}
  \mu(\mathcal{A}_i)=\left\{\begin{array}{ll}
       \max\left\{\left\lceil \frac{2\omega(\mathcal{A}_i)}{\min\{\kappa_1,\pi\}}\right\rceil,1\right\},~\text{if~$\mathcal{A}_i$~is~an~interval};\\[10pt]
       \left|\mathcal{A}_i\right|,~\text{if~$\mathcal{A}_i$~is~a~finitely~discrete~set},
    \end{array}\right.
  \end{equation} and $\kappa_1$ and $\kappa_2$ are the constants defined in \eqref{eq32} and \eqref{eq35}, respectively.

%~with~a~finite~number~of~elements
%$N(\mathcal{A}_i)$ is equal to
%\begin{equation}\label{NAI}
%
%\end{equation}
%if $\mathcal{A}_i$ is an interval and the cardinality of the set if $\mathcal{A}_i$ is a discrete set with a finite number of elements,
%$N(\mathcal{A}_i)=\left|\mathcal{A}_i\right|$ equals
%
%let
\end{theorem}
%\begin{proof}
%See Subsection F of the Appendix.
%\end{proof}

%Hence, ECSDR-BB algorithm is guaranteed to terminate in a finite number of iterations, with an $\epsilon$-optimal solution being returned.
The proof of Theorem \ref{thm-iteration} can be found in Appendix \ref{appthmiteration}. Two remarks on Theorem \ref{thm-iteration} are in order. First, from Theorem \ref{thm-iteration}, we can obtain the global convergence of the ECSDR-BB algorithm, i.e., both the sequences of the upper bounds and the lower bounds generated by the algorithm with $\epsilon=0$ converge to the optimal value of problem (CQP). In practice, we need to preselect a positive error tolerance $\epsilon$ in our proposed algorithm, as in most of iterative optimization algorithms. Second, Theorem \ref{thm-iteration} shows that the total number of iterations $K$ in \eqref{eq36} for the ECSDR-BB algorithm to return an $\epsilon$-optimal solution of problem (CQP) is exponential with respect to the number of variables $n.$ The iteration complexity of our proposed algorithm seems high at first sight. However, as will be shown in Section \ref{sec:numerical}, its practical number of iterations is actually significantly less than the worst-case bound in \eqref{eq36}. It is also worth remarking that there is no polynomial time algorithm which can globally solve the problem (unless P=NP), because the problem is NP-hard.

%Theorem 3 also provides a worst-case complexity of the ECSDR-BB algorithm. However, this complexity result only has theoretical meaning. In our
%numerical experiments, we will demonstrate that the empirical complexity of ECSDR-BB algorithm can be much lower than the theoretical complexity.

\section{Numerical Simulations}\label{sec:numerical}

%To evaluate the effectiveness and efficiency of our proposed ECSDR-BB algorithm,

In this section, we present some numerical simulation results to demonstrate the tightness of
our proposed relaxation (ECSDR) and the efficiency of our proposed ECSDR-BB algorithm for
problem (CQP). We apply the ECSDR-BB algorithm to solve three optimization problems arising from signal processing applications
introduced in {Section \ref{sec:application}}, i.e., the MIMO detection problem \eqref{MIMO-Problem}, the unimodular radar code design problem \eqref{Radar-Problem}, and the virtual beamforming design problem \eqref{VB}. All of these three problems are special cases of problem (CQP) but they have different characteristics, e.g., the MIMO detection problem has discrete argument constraints {and unit-modulus constraints}; the unimodular radar code design problem has continuous argument constraints {(but not equal $[0,2\pi]$) and unit-modulus constraints}; and the virtual beamforming design problem has continue argument constraints {(and equal $[0,2\pi]$) and interval modulus constraints.} %(i.e., $\sqrt{P_i}=u_i>\ell_i=0$ for all $i$)}.
%has interval modulus constraints. %The main purpose of our simulations is the following two:
All of the experiments are implemented in MATLAB, with $\mathrm{SeDuMi}$ \cite{Sedumi} being used to solve semidefinite programs (SDPs).
{{All the algorithms are run on a PC with Intel Core i7-2600 (3.40GHz) and 4GB memory.}}
{In all experiments, the error tolerance of
the proposed ECSDR-BB algorithm is set to be $\epsilon=10^{-4}$.
%, and the error tolerance of SeDuMi is set to be $10^{-8}$.}}

%In this section, we evaluate the effectiveness of the proposed algorithms in two aspects. On one hand, we study
%the tightness of the enhanced relaxation (ECSDR). On the other hand, we test the efficiency of the ECSDR-BB algorithm.
%
%Our numerical experiments are carried out on three
%typical problems introduced in Section I-A, including the MIMO channel detection problem, the radar code design problem, and the
%virtual beamforming design problem. The above three problems have quite different structures:
%All experiments are implemented in MATLAB, using SeDuMi
%\cite{Sedumi} to solve semidefinite programming problems. The error tolerance in all
%experiments is set to $\epsilon=10^{-4}$.

\subsection{Numerical Results of MIMO Detection}\label{sec:numericalmimo}

We generate the instances of the MIMO detection problem \eqref{MIMO-Problem} as {in \cite{Chan,Damen}}: we first generate $\bH$ according to the standard complex Gaussian distribution; then we generate a complex vector $\bx^{\ast}$ with $\left|x_i^\ast\right|=1$ and $\arg\left(x^\ast_i\right)$ being uniformly chosen from the discrete set $\mathcal{A}_i$ for all $i=1,2,\ldots,n;$ finally we set $\br=\bH\bx^\ast+\sigma \bv,$ where $\bv\in\mathbb{C}^n$ is a Gaussian noise obeying the standard complex Gaussian distribution and $\sigma$ is a parameter which controls the SNR determined by $\text{SNR}=10\log_{10}\left({\|\bH\bx^\ast\|_2^2}/{\sigma^2n}\right).$ In our simulations, $\mathcal{A}_i$ is either $\left\{0, \pi/2,\pi,3\pi/2\right\}$ (i.e., QPSK) or $\left\{0, \pi/4, \pi/2,3\pi/4, \pi,5\pi/4,3\pi/2,7\pi/4\right\}$ (i.e., 8-PSK).

For each setup, we generate {$50$} problem instances and apply our proposed algorithm (i.e., Algorithm 1) to solve them. All results in this subsection are obtained by averaging over the {$50$} generated instances.
Numerical results are summarized in Tables I and II, where ``ObjVal'' denotes the objective value returned by the proposed ECSDR-BB algorithm; ``LBdE'' and ``LBdC'' denote
the lower bounds (i.e., the optimal objective values) returned by two relaxations (ECSDR) and (CSDR) (for problem \eqref{MIMO-Problem}), respectively;
``CldGap'' denotes (LBdE$-$LBdC)/(ObjVal$-$LBdC) $\times 100\%$, which measures how much of the gap in (CSDR) is closed by (ECSDR);  ``Time" and ``\# Iter"  denote the CPU time and the number of iterations of the ECSDR-BB algorithm for solving problem (CQP); and ``TimeE" and ``TimeC" denote the CPU time of solving relaxations (ECSDR) and (CSDR), respectively. Notice that ``CldGap'' defined in the above must be in $[0,1]$ due to the fact $\textrm{ObjVal} \geq \textrm{LBdE} \geq \textrm{LBdC}.$

%The objective values returned by the ECSDR-BB algorithm (i.e., $\nu(\textrm{BB})$), the (ECSDR) bound (i.e., Bound-E), and the (CSDR) bound (i.e., Bound-C) are listed in Table I.
%%Each row lists the average results over ten instances.
%To evaluate the effectiveness of the valid inequalities in (ECSDR), the column ``Closed" lists the percentage of gap being closed by
%(ECSDR), which is defined as
%$$\left| \frac{ \nu(\textrm{ECSDR})-\nu(\textrm{CSDR})}{\nu(\textrm{ECSDR-BB})-\nu(\textrm{CSDR})}\right|,$$
%where $\nu(\cdot)$ denotes the average objective value (over ten instances) obtained by the method $(\cdot)$.
%The average CPU time  and the number of iterations are listed in Table II, where
%%The columns  is the average number of iterations of the ECSDR-BB algorithm.

\begin{table}%[htb]
\centering
\caption{Objective values and lower bounds of MIMO detection problem \eqref{MIMO-Problem}}
\begin{tabular}{c|c|rrr|r}
\hline
%\multicolumn{3}{|c|}{Instance}   & \multicolumn{3}{c|}{ECSDR-BB}  & \multicolumn{3}{c|}{SDP Approximation}  \\
$(m,n,M)$ &SNR &ObjVal  &LBdE &LBdC &CldGap \\
\hline
(15,\,10,\,4) &25  &0.954 &0.954 &0.662 &100.0 \% \\
(15,\,10,\,4) &20  &3.118 &3.102 &2.124 &98.4 \% \\
(15,\,10,\,4) &15  &9.809 &9.575 &6.395 &93.1 \% \\
(15,\,10,\,4) &10  &30.093 &28.064 &21.124 &77.4 \% \\
(15,\,10,\,4) &5  &85.843 &72.824 &56.006 &56.4 \% \\
\hline
(15,\,10,\,8) &25  &0.960 &0.953 &0.669 &97.6 \% \\
(15,\,10,\,8) &20  &3.082 &2.972 &2.034 &89.6 \% \\
(15,\,10,\,8) &15  &9.712 &8.678 &6.605 &66.7 \% \\
(15,\,10,\,8) &10  &28.858 &24.333 &20.351 &46.8 \% \\
(15,\,10,\,8) &5  &78.708 &71.169 &65.251 &44.0 \% \\
\hline
(30,\,20,\,4) &25  &3.638 &3.638 &2.366 &100.0 \% \\
(30,\,20,\,4) &20  &12.489 &12.464 &8.376 &99.4 \% \\
(30,\,20,\,4) &15  &38.064 &36.335 &24.977 &86.8 \% \\
(30,\,20,\,4) &10  &119.411 &106.786 &79.083 &68.7 \% \\
(30,\,20,\,4) &5  &348.242 &288.539 &229.577 &49.7 \% \\
\hline
(30,\,20,\,8) &25  &3.851 &3.764 &2.544 &93.3 \% \\
(30,\,20,\,8) &20  &11.765 &11.084 &7.783 &82.9 \% \\
(30,\,20,\,8) &15  &37.940 &32.936 &26.065 &57.9 \% \\
(30,\,20,\,8) &10  &115.252 &91.510 &78.141 &36.0 \% \\
(30,\,20,\,8) &5  &285.485 &252.054 &236.121 &32.3 \% \\
\hline
\end{tabular}
\end{table}

\begin{table}%[htb]
\centering
\caption{CPU time (in Seconds) and number of iterations for solving MIMO detection problem \eqref{MIMO-Problem}}
\begin{tabular}{c|c|rr|rr}
\hline
%\multicolumn{3}{|c|}{Instance}   & \multicolumn{3}{c|}{ECSDR-BB}  & \multicolumn{3}{c|}{SDP Approximation}  \\
$(m,n,M)$ &SNR &\# Iter &Time &TimeE &TimeC \\
\hline
(15,\,10,\,4) &25 &1.0 &0.09  &0.09  &0.07\\
(15,\,10,\,4) &20 &1.3 &0.14  &0.09  &0.07\\
(15,\,10,\,4) &15 &2.3 &0.29  &0.08  &0.07\\
(15,\,10,\,4) &10 &3.8 &0.53  &0.08  &0.07\\
(15,\,10,\,4) &5 &9.8 &1.36  &0.07  &0.07\\
\hline
(15,\,10,\,8) &25 &1.6 &0.23  &0.11  &0.07\\
(15,\,10,\,8) &20 &3.1 &0.47  &0.10  &0.07\\
(15,\,10,\,8) &15 &6.5 &1.04  &0.09  &0.07\\
(15,\,10,\,8) &10 &13.3 &2.08  &0.09  &0.07\\
(15,\,10,\,8) &5 &23.1 &3.49  &0.08  &0.06\\
\hline
(30,\,20,\,4) &25 &1.0 &0.18  &0.18  &0.12\\
(30,\,20,\,4) &20 &1.4 &0.30  &0.18  &0.12\\
(30,\,20,\,4) &15 &4.0 &1.03  &0.15  &0.12\\
(30,\,20,\,4) &10 &8.2 &2.14  &0.15  &0.12\\
(30,\,20,\,4) &5 &40.1 &9.25  &0.13  &0.11\\
\hline
(30,\,20,\,8) &25 &3.1 &0.96  &0.19  &0.12\\
(30,\,20,\,8) &20 &7.3 &2.32  &0.18  &0.11\\
(30,\,20,\,8) &15 &15.2 &4.72  &0.18  &0.12\\
(30,\,20,\,8) &10 &77.8 &21.02  &0.15  &0.11\\
(30,\,20,\,8) &5 &161.3 &42.36  &0.15  &0.11\\
\hline
\end{tabular}
\end{table}

We first compare the two relaxations (ECSDR) and (CSDR) in terms of their tightness and computational efficiency.
We can see from Table I that relaxation (ECSDR)
is generally much tighter than (CSDR). In particular, in cases of $M=4$ and $\text{SNR}=25$, our proposed relaxation (ECSDR)
is exact, %\footnote{\rev{This phenomenon can be explained by the analysis results in \cite{lu2019tightness}, which show that (ECSDR), when applied to the MIMO detection problem, will be tight with overwhelming high probability if the channel matrix and the noise obey the Gaussian distribution and the SNR is sufficiently high; please see \cite[Theorems 4.4 and 4.5]{lu2019tightness} for more details.}},
i.e., the optimal solution and the optimal objective value of relaxation (ECSDR) are equal to that of the original problem; %and the solution to (ECSDR) is also a solution of the original problem;
in cases of SNR$\geq 15$, our proposed relaxation (ECSDR) narrows down over $50\%$ of the gap due to (CSDR);
and in all cases, LBdE is generally larger than LBdC and more than $30\%$ of the gap in (CSDR) is closed by (ECSDR). These results clearly show
that the new envelope constraints added in (ECSDR) (as compared to (CSDR)) are indeed very useful to reduce the relaxation gap of (CSDR).
From Table II, we can observe that the two relaxations have similar computational efficiency. %The
%(average) CPU time of solving (CSDR) is about 0.06--0.12 seconds and the one of solving (ECSDR)
%is about {0.07--0.19} seconds.
In fact, since the number of constraints in (ECSDR) is larger than that of  (CSDR), %with more convex constraints being added,
the CPU time of solving (ECSDR)
is generally larger than that of solving (CSDR). However, the CPU time of solving (ECSDR) is not much larger than that of solving (CSDR). This is because the constraints added in (ECSDR) (compared to (CSDR)) are all ``simple'' linear constraints. Based on the above analysis and numerical results, we conclude that (ECSDR) generally is much tighter
than (CSDR) but solving (ECSDR) takes only slightly more CPU time than solving (CSDR).
%the two relaxations has comparable computational efficiency.

%the CPU time of solving (ECSDR) is not much larger than
% only needs slightly more computational time than (CSDR).
%
%All above actually can
%
%(less than two times).
%, which can be seen from Table

%Second, we evaluate the efficiency of ECSDR-BB algorithm.
We can also see from Table II that our proposed ECSDR-BB algorithm can solve all problem instances within
{$162$ iterations and within $43$ seconds} (on average). %In particular, our proposed algorithm
%terminates within $8$ iterations for problem instances with their SNRs being larger than or equal to $20$dB; and for problem instances with their SNRs being $25$dB and $M=4$,
%our proposed algorithm terminates within only $1$ iteration (which implies that our proposed relaxation (ECSDR) is tight in this case).
These results show that our proposed ECSDR-BB algorithm is very efficient for globally solving the MIMO detection problem. We will further compare the efficiency of our proposed ECSDR-BB algorithm with a specially designed algorithm for solving the MIMO detection problem in Section IV-D.

%[[[\rev{Might be better to mention the tightness result in our paper \cite{lu2017tightness} submitted to SIAM somewhere in this subsection. Let me know your idea on this!}]]]

\subsection{Numerical Results of Unimodular Radar Code Design}
We generate the instances of the unimodular radar code design problem \eqref{Radar-Problem} as {in \cite{Maio2009}}: %(see \cite{Maio2009}):
%$\bQ=\bM^{-1}\odot (\bp \bp^\dag)^*$, where
we set each entry of $\bM$ to be $M_{ij}=\rho^{|i-j|}$ with $\rho\in [0.2,0.8];$ set $N=7$ and $\bx^0=[1,1,1,-1,-1,1,-1]^\T$ (i.e., the Barker code of length $7$);
set the steering vector $\bp=\left[1,e^{\textbf{i}2\pi f_dT_r},\ldots,e^{\textbf{i}2\pi (N-1) f_dT_r}\right]^\T$ with $f_dT_r=0.15$; and set
the similarity tolerance $\delta$ in \eqref{Radar-Problem} such that
$\arccos(1-\delta^2/2)$ is equal to either $\pi/6$ or $\pi/3$. %, respectively.
%As in \cite{Maio2009}, we generate 7-dimensional test instances, with $c^0=[1,1,1,-1,-1,1,-1]^\T$ being .
%The value $\arccos(1-\epsilon^2/2)$ is set to $\pi/6$ and $\pi/3$, respectively.

For each value of the parameter $\delta$, we generate $5$ problem instances and apply our proposed algorithm (i.e., Algorithm 1) to solve them. %All results in this subsection are obtained by averaging over $10$ instances.
Numerical results on all $10$ problem instances are summarized in Tables III and IV, where ``ID'' denotes the IDs of the corresponding problem instances, $\omega(\mathcal{A}_{i})$ denotes the width of set $\mathcal{A}_{i}$ (cf. \eqref{WAi}), ``UBdE'' and ``UBdC'' denote
the upper bounds\footnote{Recall that problem \eqref{Radar-Problem} is a maximization problem and thus the optimal values of the corresponding relaxations are upper bounds of its optimal value.} (i.e., the optimal objective values) returned by two relaxations (ECSDR) and (CSDR) (for problem \eqref{Radar-Problem}), respectively,
``CldGap'' denotes (UBdC$-$UBdE)/(UBdC$-$ObjVal) $\times 100\%$, and all the others have the same meanings as that in Tables I and II. By the definition of $\omega(\mathcal{A}_{i})$ (cf. \eqref{WAi}), we have $\omega(\mathcal{A}_{i})=\pi/3$ if $\arccos(1-\delta^2/2)=\pi/6$ and $\omega(\mathcal{A}_{i})=2\pi/3$ if $\arccos(1-\delta^2/2)=\pi/3.$

%ten are for the case $\arccos(1-\epsilon^2/2)=\pi/6$, and the other ten are for the case $\arccos(1-\epsilon^2/2)=\pi/3$.
%Similar to the experiments in Section IV-A, we solve all the test
%instances by using ECSDR-BB algorithm to obtain their optimal values, and then compute the (ECSDR) and (CSDR) based lower bounds
%for comparison. The objective values computed by different algorithms are listed in Table III, and the
%CPU time of different algorithms are listed in Table IV, in which the column ``Time 1", ``Time 2", ``Time 3" and ``Enum"
%have the same meanings as those in Table II.

\begin{table}
\centering
\caption{Objective values and upper bounds of unimodular radar code design problem \eqref{Radar-Problem}}
\begin{tabular}{c|c|c|cc|c}
%\toprule
\hline
ID &$\omega(\mathcal{A}_{i})$  &ObjVal & UBdE &UBdC &CldGap\\
\hline
1 &$\pi/3$ &73.76 &74.47 &140.93 &99.0\% \\
2 &$\pi/3$ &11.74 &11.88 &14.63 &95.0\% \\
3 &$\pi/3$ &10.73 &10.88 &13.97 &95.3\% \\
4 &$\pi/3$ &10.55 &10.67 &14.00 &96.4\% \\
5 &$\pi/3$ &12.23 &12.24 &14.06 &99.4\% \\
%6 &$\pi/3$ &48.08 &49.43 &104.77 &97.6\% \\
%7 &$\pi/3$ &11.35 &11.35 &13.75 &99.9\% \\
%8 &$\pi/3$ &15.26 &15.48 &22.08 &96.8\% \\
%9 &$\pi/3$ &16.88 &17.34 &25.18 &94.4\% \\
%10 &$\pi/3$ &14.80 &15.02 &23.13 &97.3\% \\
\hline
6 &$2\pi/3$ &20.55 &22.12 &24.13 &56.1\% \\
7 &$2\pi/3$ &35.02 &40.41 &48.88 &61.1\% \\
8 &$2\pi/3$ &12.30 &12.94 &13.76 &56.0\% \\
9 &$2\pi/3$ &16.91 &17.59 &19.09 &68.8\% \\
10 &$2\pi/3$ &22.00 &23.11 &25.02 &63.1\% \\
%16 &$2\pi/3$ &13.72 &14.35 &15.06 &53.6\% \\
%17 &$2\pi/3$ &18.08 &18.40 &18.54 &31.5\% \\
%18 &$2\pi/3$ &20.92 &22.42 &23.57 &43.4\% \\
%19 &$2\pi/3$ &14.74 &16.33 &17.56 &43.6\% \\
%20 &$2\pi/3$ &38.85 &42.79 &44.80 &33.7\% \\
\hline
\end{tabular}
\end{table}

%Note that the problem is to
%maximize to objective function, thus the bounds are upper bounds rather than lower bounds.

\begin{table}
\centering
\caption{CPU time (in Seconds) and number of iterations for solving unimodular radar code design problem \eqref{Radar-Problem}}
\begin{tabular}{c|c|rr|rr}
\hline
ID &$\omega(\mathcal{A}_{i})$   &\# Iter &Time  &TimeE &TimeC \\
\hline
1 &$\pi/3$ &19 &2.09 &0.05 &0.03 \\
2 &$\pi/3$ &12 &1.17 &0.05 &0.04 \\
3 &$\pi/3$ &11 &1.07 &0.05 &0.04 \\
4 &$\pi/3$ &6 &0.53 &0.05 &0.03 \\
5 &$\pi/3$ &2 &0.15 &0.05 &0.04 \\
%6 &$\pi/3$ &25 &2.75 &0.05 &0.03 \\
%7 &$\pi/3$ &2 &0.14 &0.05 &0.04 \\
%8 &$\pi/3$ &4 &0.36 &0.05 &0.03 \\
%9 &$\pi/3$ &18 &1.75 &0.05 &0.04 \\
%10 &$\pi/3$ &3 &0.26 &0.05 &0.04 \\
\hline
6 &$2\pi/3$ &24 &2.45 &0.05 &0.03 \\
7 &$2\pi/3$ &8 &0.78 &0.05 &0.04 \\
8 &$2\pi/3$ &20 &1.93 &0.05 &0.03 \\
9 &$2\pi/3$ &9 &0.83 &0.05 &0.03 \\
10 &$2\pi/3$ &11 &1.06 &0.05 &0.04 \\
%16 &$2\pi/3$ &7 &0.64 &0.05 &0.04 \\
%17 &$2\pi/3$ &15 &1.48 &0.05 &0.04 \\
%18 &$2\pi/3$ &17 &1.68 &0.05 &0.04 \\
%19 &$2\pi/3$ &31 &3.09 &0.05 &0.04 \\
%20 &$2\pi/3$ &24 &2.40 &0.05 &0.04 \\
\hline
\end{tabular}
\end{table}

We can observe and conclude from the results listed in Tables III and IV that:
\begin{itemize}
  \item [1)] (ECSDR) is generally much tighter than (CSDR), especially when
$\omega(\mathcal{A}_{i})$ is small. %In particular, when $\omega(\mathcal{A}_{i})=\pi/3$, more than $95\%$ of the gap
%in (CSDR) is closed by (ECSDR), while when $\omega(\mathcal{A}_{i})=2\pi/3$, $50\%$--$70\%$ of the gap is closed.
This is consistent with the analysis in Proposition~\ref{lemma-width}: since $r_i=1$ in the unimodular radar code design problem \eqref{Radar-Problem}, then
the inequality in \eqref{eq10} reduces to $1\geq \left|x_i\right|\geq \cos\left({\omega(\mathcal{A}_{i})}/{2}\right),$
%%can be interpreted by Lemma 1 as follows:
%the inequality $$r_i\geq \left|x_i\right|\geq r_i\cos\left({\omega(\mathcal{A}_{i})}/{2}\right),$$
%where ,
which shows that a smaller $\omega(\mathcal{A}_{i})$ generally leads to a smaller gap $1 -\left|x_i\right|.$
\item [2)]  Solving (ECSDR) takes slightly more CPU time than solving (CSDR).
\item [3)] Our proposed ECSDR-BB algorithm is able to efficiently solve all generated problem instances
within satisfactory computational time (i.e., less than {$3$} seconds) and within a relatively small number of iterations (i.e., less than $24$ iterations).
\end{itemize} %is also small.

\subsection{Numerical Results of Virtual Beamforming Design}\label{sec:numericalvb}
We generate the instances of the virtual beamforming design problem \eqref{VB}:
we set %$m=n\in\{10,15,20\}$ and
$P_i=1$ for all $i=1,2,\ldots,n$ and generate $\bh_j$ for all $j=1,2,,\ldots,m$ according to the standard complex Gaussian distribution (as {in \cite[Section 4.3.5]{Palomar}}).
In our simulations, we set $m\in\{5,10,15\}$ and $n\in\{5,10,15,20\}$ (and thus there are in total $12$ different pairs of $(m,\,n)$).
For each pair of $(m,\,n)$, we generate {$50$} instances and apply our proposed ECSDR-BB algorithm to solve them.
All results in this subsection are obtained by averaging over the {$50$} instances and the obtained results are summarized in Tables V and VI.

%We assume that the channel vectors $\bh_1,\ldots,\bh_m$ follow the complex Gaussian distribution.
%We test the performance of the ECSDR-BB algorithm on 120 test instances with
%$m\in\{5,10,15\}$, $n\in\{5,10,15,20\}$ and $P=1$. The vectors $\bh_1,\bh_2,\ldots,\bh_m$ are randomly sampled from the $n$-dimensional complex Gaussian distribution $\mathcal{C}\mathcal{N}(0, I_n)$. Similar to the previous two subsections, we list the related results of ECSDR-BB algorithm,
%(CSDR) relaxation and (ECSDR) relaxation in Tables V and VI. Each row of these tables corresponds to the average results
%over ten instances.

\begin{table}%[htb]
\centering
\caption{Objective values and upper bounds of virtual beamforming design problem \eqref{VB}}
%The problem is to maximize the objective function, thus
%the listed bounds are upper bounds.}
\begin{tabular}{c|rrr}
\hline
$(m,n)$  &ObjVal &UBdE &UBdC \\
\hline
(5,\,5)  &108.837 &108.858  &108.858 \\
(10,\,5)  &187.625 &187.647  &187.647 \\
(15,\,5)  &259.929 &260.127  &260.127 \\
(5,\,10)  &364.247 &364.746  &364.746 \\
(10,\,10)  &534.053 &535.807  &535.807 \\
(15,\,10)  &701.893 &703.075  &703.075 \\
(5,\,15)  &696.457 &699.092  &699.092 \\
(10,\,15)  &1030.088 &1037.109  &1037.109 \\
(15,\,15)  &1320.666 &1327.523  &1327.523 \\
(5,\,20)  &1154.460 &1163.583  &1163.583 \\
(10,\,20)  &1619.981 &1633.517  &1633.517 \\
(15,\,20)  &2039.053 &2057.837  &2057.837 \\
\hline
\end{tabular}
\end{table}

\begin{table}%[htb]
\centering
\caption{CPU time (in Seconds) and number of iterations for solving virtual beamforming design problem \eqref{VB}}
\begin{tabular}{c|rr|rr}
\hline
$(m,n)$  &\# Iter &Time  &TimeE &TimeC \\
\hline
(5,\,5)  &1.6 &0.19  &0.06 &0.06 \\
(10,\,5)  &1.5 &0.17  &0.06 &0.06 \\
(15,\,5)  &3.1 &0.45  &0.06 &0.06 \\
(5,\,10)  &11.4 &3.01  &0.06 &0.06 \\
(10,\,10)  &20.6 &5.68  &0.06 &0.06 \\
(15,\,10)  &14.4 &3.89  &0.06 &0.06 \\
(5,\,15)  &94.2 &38.92  &0.06 &0.06 \\
(10,\,15)  &157.1 &58.51  &0.06 &0.06 \\
(15,\,15)  &161.4 &58.25  &0.06 &0.06 \\
(5,\,20)  &601.4 &267.86  &0.06 &0.06 \\
(10,\,20)  &534.8 &237.52  &0.06 &0.06 \\
(15,\,20)  &541.7 &242.19  &0.06 &0.06 \\
\hline
\end{tabular}
\end{table}

%We compare the bounds of the (CSDR) relaxation and (ECSDR) relaxation with the global optimal value in Table V.
{Since the set $\mathcal{A}_i=[0,2\pi]$ for all $i=1,2,\ldots,n$ in problem \eqref{VB}, relaxation (ECSDR) is equivalent to (CSDR), as
discussed at the end of Section II and as demonstrated {and verified} in Table V.
%In particular, both of the two relaxations are tight in all cases of $n=5.$
%Although relaxation (ECSDR) is not tighter than (CSDR) in the special case (where $\mathcal{A}_i=[0,2\pi]$ for all $i=1,2,\ldots,n$) as in problem \eqref{VB},
However, (ECSDR) still plays an important role of globally solving problem \eqref{VB} in the ECSDR-BB
algorithm, where %as the set $\mathcal{A}_i$ is recursively partitioned into smaller subsets,
the quality of (ECSDR) defined over the recursively partitioned subsets of $\mathcal{A}_i$ becomes better and better.}
%generating the lower bounds in the ECSDR-BB
%algorithm (for globally solving problem \eqref{VB}), where as the set $\mathcal{A}_i$ is recursively partitioned into smaller subsets, the quality of relaxation (ECSDR) defined over the subsets will become better and better.}
%relaxation becomes effective for lower bounding.

%. Hence, as demonstrated in Table 5, the bounds of the two relaxations are always
%the same. In all the cases of $n=5$, the (CSDR) relaxation and (ECSDR) relaxation are both very tight, and the gaps
%between these lower bounds and the global optimal value is zero. As $n$ increases to $10$, the two relaxations
%are not tight, and the gaps become nonzero. Although the (ECSDR) relaxation seems not tighter than (CSDR) relaxation
%in the case of $\mathcal{A}_i=[0,2\pi]$, $i=1,2,\ldots,n$, it is still an important lower bounding component in the ECSDR-BB
%algorithm, in which the set $\mathcal{A}_i$ is partitioned to smaller subsets, and then the (ECSDR) relaxation becomes effective for lower bounding.

%The CPU time results listed
{As shown in Table VI, our proposed ECSDR-BB algorithm is quite efficient for solving small-scale problem instances (e.g., with $n\leq 10$).
%Since relaxation (ECSDR) is always tight for the instances with $n=5$, the ECSDR-BB algorithm terminates in
%only one iteration for these instances.
%The ECSDR-BB algorithm can solve problem instances with $n=5,$ $n=10,$ and $n=15$ within $1$ second, $6$ seconds, and $59$ seconds (on average), respectively.
Moreover, the ECSDR-BB algorithm is able to solve all of problem instances within $602$ iterations and within $268$ seconds (on average).
These results show the high efficiency of our proposed algorithm in solving the virtual beamforming design problem \eqref{VB}.
%that most of problem instances can be solved efficiently by our proposed algorithm.~% in practical computing time.
By using our proposed algorithm as the benchmark, we can see that the two relaxations (for the original problem \eqref{VB}) are generally not tight for the virtual beamforming design problem but the relaxation gaps are generally very small. We shall further compare the efficiency of our proposed ECSDR-BB algorithm with the general-purpose global optimization solver for solving the virtual beamforming design problem \eqref{VB} in Section IV-E.}

%
%Moreover,
%, all the 10-dimensional cases can be solved by
%, and all the 15-dimensional cases can be solved within 30 seconds (on average).
%In these cases, the CPU time is practical for real applications. Moreover, for the 20-dimensional instances, ECSDR-BB
%algorithm may need tens of minutes, which might not be efficient enough for all cases, but is still acceptable
%for generating benchmark solutions for testing approximation algorithms.

\subsection{Comparison of ECSDR-BB with SD for MIMO Detection}

In this subsection, we compare our proposed ECSDR-BB algorithm with the state-of-the-art tailored global algorithm called sphere decoder\footnote{The code of the SD algorithm is downloadable from  https://ww2.mathworks.cn/matlabcentral/fileexchange/22890-sphere-decoder-for-mimo-systems. We have made some modifications on the above downloaded code to improve its efficiency by adopting the techniques proposed in \cite{Chan}.}%, which indeed significantly improves the computational efficiency of the original code
~(SD) \cite{Damen} for solving the MIMO detection problem \eqref{MIMO-Problem}. %The SD algorithm is a tailored global algorithm for the MIMO detection problem \cite{Damen}.
%As far as we know, there is no previous global algorithm that is specially designed for solving (CQP), except the case of
%MIMO detection problem. For the MIMO detection problem, the most well-known global algorithm is the sphere decoder (see, e.g., \cite{Damen}).
%In this section, we compare the efficiency of the proposed ECSDR-BB algorithm with the classical sphere decoder algorithm.
{{To compare the two algorithms, we generate problem instances with $M=8$ and $M=4$, and different $(m,\,n)$ and different SNRs.}}
In each setup, we generate {$50$} problem instances and apply the two algorithms to solve them.
{{Numerical results of the average and worst-case CPU time of the 50 instances are summarized in Table VII and Table VIII.}}

%thirty instances are randomly generated, and solved by using the ECSDR-BB algorithm and
%sphere decoder. The improved techniques
%proposed in \cite{Chan} have been applied to speed-up the sphere decoder.

%\begin{table}[htb]
%\centering
%\caption{ECSDR-BB versus sphere decoder (SD) on CPU time for solving MIMO detection problem.}
%\begin{tabular}{cc|rr|rr}
%\hline
%\multicolumn{2}{c|}{Configure}  & \multicolumn{2}{c|}{Average Performance}  & \multicolumn{2}{c}{Worst-Case Performance}  \\
%$(m,n,M)$ &SNR &ECSDR-BB &SD &ECSDR-BB &SD \\
%\hline
%(20,10,8) &25  &0.127 &0.144 &0.510 &0.532  \\
%(20,10,8) &20  &0.199 &0.227 &0.949 &1.107  \\
%(20,10,8) &15  &0.708 &0.244 &1.589 &1.703  \\
%(20,10,8) &10  &1.704 &0.176 &5.231 &0.733  \\
%(20,10,8) &5  &3.625 &0.578 &19.550 &6.074  \\
%\hline
%(30,15,8) &25  &0.120 &1.906 &0.156 &7.790  \\
%(30,15,8) &20  &0.256 &2.138 &1.463 &10.069  \\
%(30,15,8) &15  &1.170 &2.076 &2.218 &17.641  \\
%(30,15,8) &10  &3.572 &1.959 &7.985 &6.654  \\
%(30,15,8) &5  &10.200 &3.839 &40.017 &20.113  \\
%\hline
%(40,20,8) &25  &0.721 &12.651 &5.058 &35.595  \\
%(40,20,8) &20  &0.489 &17.683 &1.365 &106.973  \\
%(40,20,8) &15  &2.424 &41.757 &6.180 &668.101  \\
%(40,20,8) &10  &8.228 &46.777 &30.577 &807.337  \\
%(40,20,8) &5  &23.002 &213.456 &129.798 &4009.244  \\
%\hline
%\end{tabular}
%\end{table}

\begin{table}%[htb]
\centering
\caption{CPU time of ECSDR-BB and SD {{\cite{Chan}}} for solving MIMO detection problem \eqref{MIMO-Problem} {with $M=8$}.}
\begin{tabular}{cc|rr|rr}
\hline
\multicolumn{2}{c|}{Problem Setup}  & \multicolumn{2}{c|}{Average Performance}  & \multicolumn{2}{c}{Worst-Case Performance}  \\
$(m,n,M)$ &SNR &ECSDR-BB &SD &ECSDR-BB &SD \\
\hline
(30,\,20,\,8) &25  &0.756 &0.006 &3.199 &0.155  \\
(30,\,20,\,8) &20  &2.289 &0.004 &4.363 &0.006  \\
(30,\,20,\,8) &15  &4.533 &0.020 &8.818 &0.071  \\
(30,\,20,\,8) &10  &22.699 &0.991 &130.547 &4.527  \\
(30,\,20,\,8) &5  &45.675 &105.955 &142.117 &2004.779  \\
\hline
(24,\,20,\,8) &25  &2.954 &0.004 &6.471 &0.005  \\
(24,\,20,\,8) &20  &4.685 &0.012 &8.385 &0.091  \\
(24,\,20,\,8) &15  &8.313 &0.325 &27.128 &1.781  \\
(24,\,20,\,8) &10  &52.040 &25.201 &183.207 &627.656  \\
(24,\,20,\,8) &5  &70.287 &566.220 &196.828 &10825.587  \\
\hline
(20,\,20,\,8) &25  &5.477 &5.744 &9.313 &246.148  \\
(20,\,20,\,8) &20  &7.136 &17.862 &14.032 &320.032  \\
(20,\,20,\,8) &15  &14.844 &86.448 &43.206 &1922.981  \\
(20,\,20,\,8) &10  &102.886 &281.456 &433.586 &5543.505  \\
(20,\,20,\,8) &5  &139.449 &1836.163 &571.820 &14263.330  \\
\hline
\end{tabular}
\end{table}

\begin{table}%[htb]
\centering
\caption{CPU time of ECSDR-BB and SD {{\cite{Chan}}} for solving MIMO detection problem \eqref{MIMO-Problem} {with $M=4$}.}
\begin{tabular}{cc|rr|rr}
\hline
\multicolumn{2}{c|}{Problem Setup}  & \multicolumn{2}{c|}{Average Performance}  & \multicolumn{2}{c}{Worst-Case Performance}  \\
$(m,n,M)$ &SNR &ECSDR-BB &SD &ECSDR-BB &SD \\
\hline
(24,\,20,\,4) &15  &0.634 &0.003 &0.957  &0.027 \\
(36,\,30,\,4) &15  &2.046 &0.008 &2.819  &0.041 \\
(48,\,40,\,4) &15  &4.912 &0.421 &6.122  &8.151 \\
(60,\,50,\,4) &15  &11.799 &0.304 &14.204  &5.223 \\
\hline
(24,\,20,\,4) &10  &1.085 &0.045 &3.164  &0.129 \\
(36,\,30,\,4) &10  &4.216 &0.922 &12.135  &5.643 \\
(48,\,40,\,4) &10  &14.780 &8.827 &26.824  &46.277 \\
(60,\,50,\,4) &10  &40.575 &106.892 &145.908  &328.158 \\
\hline
(20,\,20,\,4) &15  &0.835 &0.386 &1.149  &1.934 \\
(22,\,22,\,4) &15  &1.193 &0.398 &1.573  &2.204 \\
(24,\,24,\,4) &15  &1.331 &3.948 &1.788  &20.830 \\
(26,\,26,\,4) &15  &1.753 &7.853 &3.200  &38.066 \\
(28,\,28,\,4) &15  &1.958 &18.704 &2.959  &210.623 \\
\hline
(20,\,20,\,4) &10  &1.857 &0.653 &4.509  &2.962 \\
(22,\,22,\,4) &10  &3.322 &3.270 &8.068  &12.719 \\
(24,\,24,\,4) &10  &2.634 &9.061 &5.244  &34.818 \\
(26,\,26,\,4) &10  &6.523 &20.186 &17.873  &188.251 \\
(28,\,28,\,4) &10  &7.872 &34.918 &19.456  &224.493 \\
\hline
\end{tabular}
\end{table}

%From the listed results, we discover that
{{For the case where $M=8$}},
we can observe, from Table VII, that our proposed ECSDR-BB algorithm is not as efficient as SD for problems where $n=20,~m\geq 24,$ and $\textrm{SNR}\geq 10$.
{However, the ECSDR-BB algorithm performs faster than the SD algorithm over the whole range of tested SNRs in the case where $(m,n)=(20,20)$.
Moreover, the ECSDR-BB algorithm becomes much faster than the SD algorithm in the low SNR case where $\textrm{SNR}=5.$}
%It is widely believed that the case where the number of inputs and outputs is equal to each other is the hardest case for the MIMO detection problem.
The reasons behind the above simulation results might be as follows. In the case where $m$ is much larger than
$n$, the matrix $\bH^{\dagger} \bH /m$ tends to be close to the $n\times n$ identity matrix $\bI_n$. %[[[\rev{Please give a reference here.}]]]
In this case,
since $\bH^{\dagger} \bH \approx m \bI_n$, %and $\bH^{\dagger} \approx m \bH^{-1} $, [[[\rev{The inverse here means pseudo-inverse???}]]]
we have
\begin{align}
\frac{1}{2}\left\|\bH\bx-\br\right\|_2^2 \approx  \frac{m}{2}\bx^{\dagger}\bx-\mathrm{Re}(\bx^{\dagger}\bH^{\dagger}\br)+\frac{1}{2}\br^{\dagger}\br,
\end{align}
and thus the global solution of the original problem is very close to $\mathrm{Scale}(\bH^{\dagger}\br,\be_n)$, where $\mathrm{Scale}(\cdot,\cdot)$ is defined in \eqref{scale} and $\be_n$ is the all-one vector of dimension $n$. Based on the above observation, a SD variant in \cite{Chan}, which applies a depth first search and
selects the node according to an increasing distance from $\bH^{\dagger}\br$ at each iteration, achieves a very high efficiency.
However, for the cases with a fixed $n$, as $m$ decreases, %the condition number of the matrix  $\bH^{\dagger} \bH$ becomes large,
$\mathrm{Scale}(\bH^{\dagger}\br,\be_n)$ might not be a good estimator of the global solution of the original problem and hence the performance of the SD variant degrades very quickly, especially in the low SNR cases. % where the SNRs are very low.

{{We can make the same observation on the comparison of the ECSDR-BB and SD algorithms from the results in Table VIII where $M=4$. In particular, the SD algorithm performs better than the ECSDR-BB algorithm in the easy cases, whereas the performance of the SD algorithm sharply degrades when the problems become hard and the ECSDR-BB algorithm performs much better than the SD algorithm in the hard cases.}}% (e.g., $m/n=1.2$ and SNR$=15$)
% (e.g., $m=n\geq 24,$ or $m=60,n=50$ for the case SNR$=10$, and $m=n\geq 26$  for the case SNR$=15$) 
%Especially, when , SD performs much better than ECSDR-BB algorithm on all cases of different $(m,n)$.
%However, for hard cases, e.g., $m=n\geq 24$ or $m=60,n=50$ for the case SNR$=10$, and $m=n\geq 26$  for the case SNR$=15$, 
%the performance of SD algorithm degrade sharply.}}

%However, as the problem becomes harder (i.e., $n$ is large, SNR is low, and/or $m/n$ equals $1$), the performance of
%SD algorithm degrade sharply (e.g., when $m=n\geq 24$ for the case SNR$=10$, when $m=n\geq 26$  for the case SNR$=15$, and
%when $m=60,n=50$ for the case SNR$=10$).}}

{In short summary, our proposed ECSDR-BB algorithm performs very well in the hard cases}, i.e., the number of inputs and outputs is equal or the SNR is low. Numerical results in Table VII {{and Table VIII}} show that our proposed ECSDR-BB algorithm exhibits a very promising performance in solving the MIMO channel detection problem \eqref{MIMO-Problem} in the hard cases. It is worth mentioning that
the SD algorithm/variant is specially designed to solve the MIMO detection problem and it seems (at least to us) not trivial to extend it to
solve more general problems while our proposed ECSDR-BB algorithm is able to solve various problems in the form of
(CQP), with continuous and/or discrete argument constraints.

%Note that the classical sphere decoder algorithm is specially designed to solve the MIMO detection problem, which is extremely
%efficient in the case that $\bH^{\dagger} \bH /m$ is close to an identity matrix. However, the sphere decoder is not
%easy to be extended to more general cases, especially to other types of (CQP) with discrete arguments. In comparison,
%the proposed ECSDR-BB algorithm is a general algorithm that can solve different types of (CQP).

\subsection{Comparison of ECSDR-BB with Baron for Virtual Beamforming Design}\label{sec:numericalbaron}
%[[[\rev{I have added some comments on Baron. How Baron works and how we can use Baron to solve our problem? Please double check and possibly add more.}]]]
In this subsection, to further demonstrate the efficiency of our proposed ECSDR-BB algorithm, we compare it with Baron {{(version 18.8.24)}} \cite{Tawarmalani}, a well-known
general-purpose global optimization solver, by applying them to solve the virtual beamforming design problem\footnote{There is no existing specially designed global algorithm for the virtual beamforming design problem \eqref{VB} that we can compare our algorithm with.} \eqref{VB}. Notice that Baron is also a branch-and-bound algorithm but it is based on linear programming relaxation. To apply Baron to solve problem \eqref{VB}, we
need to first transform the problem into a real quadratic problem by representing the real and imaginary
parts of each complex variable with two independent real variables. %\rev{as follows:
%\begin{align*}
%\min_{\by\in\mathbb{R}^{2n}} ~&~ \by^{\T}\hat \bQ \by \\
%\mbox{s.t.~~} ~&~ y_i^2+y_{n+i}^2\leq {P_i},~i=1,2,\ldots,n,
%\end{align*}
%where \begin{align*}\hat{\bQ}=\begin{bmatrix}
%\textrm{Re}(\bQ) ~&-\textrm{Im}(\bQ)\\[3pt]
%\textrm{Im}(\bQ) ~&\textrm{Re}(\bQ)\\
%\end{bmatrix},~\by=\begin{bmatrix} %\mathcal{T}(x)=
%\textrm{Re}(\bx) \\[3pt]
%\textrm{Im}(\bx)\\
%\end{bmatrix},~\bQ=- \sum_{j=1}^m \bh_j \bh_j^\dagger.
%\end{align*}}
The error tolerances in both the ECSDR-BB algorithm and Baron are set to $10^{-4}$.
In our simulations, we generate $10$ problem instances with $n=5$ and $10$ problem instances with $n=10$ as done in Section \ref{sec:numericalvb}.
%, which implements a linear relaxation
%based branch-and-bound algorithm \cite{Tawarmalani}. We use the instances in virtual beamforming design problem for comparison experiments.
%The error tolerances of both ECSDR-BB algorithm and Baron are set to $10^{-4}$.

{{Table IX}} shows the numerical results for the $10$ problem instances with $n=5.$ We can observe from the table that (ECSDR) is tight %\footnote{\rev{It follows from  \cite[Theorem 2.1]{Pataki} (also see \cite[Theorem 2.1]{lemon2016low}) that, if there are $5$ linear constraints in an SDP, then it has a solution with rank being less than or equal to $2$. This result sheds useful insight into why the SDR is generally tight in Table VIII and the algorithm terminates within a single iteration, albeit it cannot guarantee the tightness of the SDR.}}
in this case and the ECSDR-BB algorithm terminates in only one iteration within {$0.11$} {seconds}; but Baron needs significantly larger number of iterations (i.e., in the order of $1000$--$10000$) and more CPU time {(i.e., from $41$ to $432$ seconds)}.
Moreover, we have further compared the two algorithms on $10$ problem instances with $n=10$.
However, we found that Baron fails to solve most of problem instances with $(m,n)=(10,10)$ within 60 minutes (and thus the results are not listed here).
In contrast, the ECSDR-BB algorithm can successfully solve problem instances with $m\in\{5,10,15\}$ and $n=10$ within {$6$} seconds (on average), as listed in Table  VI.
These results clearly show that the specially designed ECSDR-BB algorithm achieves significantly higher efficiency on globally solving problem (CQP) than the general-purpose global optimization solver such as Baron.
%, we list numerical results for the ten 5-dimensional test instances , where the first column shows the configuration of the test instances, the second column lists the objective values, the number of iterations, and CPU time results of the ECSDR-BB algorithm, the third column lists the objective values, the number of iterations, and CPU time results of Baron.

\begin{table}%[htb]
\centering
\caption{Numerical results of ECSDR-BB and Baron for virtual beamforming design problem \eqref{VB} with $n=5$}
\begin{tabular}{c|rrr|rrr}
\hline
Setup   & \multicolumn{3}{c|}{ECSDR-BB}  & \multicolumn{3}{c}{Baron}  \\
$(m,n)$  & ObjVal &\# Iter  & Time & ObjVal & \# Iter  & Time  \\
\hline
(5,\,5)  &102.97 &1 &0.06 &102.97 &1567 &80.0 \\
(5,\,5)  &122.66 &1 &0.11 &122.66 &981 &41.9 \\
(5,\,5)  &91.40 &1 &0.11 &91.40 &1729 &62.7 \\
(5,\,5)  &120.71 &1 &0.11 &120.71 &967 &60.8 \\
(5,\,5)  &98.62 &1 &0.08 &98.62 &1705 &66.4 \\
\hline
(10,\,5)  &300.10 &1 &0.08 &300.10 &1855 &80.3 \\
(10,\,5)  &141.27 &1 &0.09 &141.27 &4071 &208.7 \\
(10,\,5)  &150.20 &1 &0.08 &150.20 &10009 &431.7 \\
(10,\,5)  &135.02 &1 &0.06 &135.02 &6999 &322.8 \\
(10,\,5)  &166.03 &1 &0.07 &166.03 &2477 &86.2 \\
\hline
\end{tabular}
\end{table}

%[[[\rev{Please double check that 0.66 in Table VIII is not a typo, or it should be 0.06?}]]]

\section{Conclusions}\label{sec:conclusion}

In this paper, we considered a class of nonconvex complex quadratic programming problems (i.e., problem (CQP)), which finds many important signal processing applications.
We first derived a new enhanced relaxation (ECSDR) (compared to the conventional relaxation (CSDR)) for problem (CQP) based on the polar coordinate representations of the complex variables. %The enhanced relaxation is obtained by carefully exploiting the special structure of the nonconvex constraints in problem (CQP)
Then we proposed a branch-and-bound global algorithm, ECSDR-BB, for solving problem (CQP) based on the newly derived relaxation. To the best of our knowledge, our proposed ECSDR-BB algorithm is the first tailored algorithm for problem (CQP) which is guaranteed to find the global solution of the problem (within any given error tolerance). We applied our proposed ECSDR-BB algorithm for solving the MIMO detection problem, the unimodular radar code design problem, and the virtual beamforming design problem, and
our simulation results show the high effectiveness of our proposed enhanced relaxation (ECSDR) and the high efficiency of our proposed ECSDR-BB algorithm.
In particular, our proposed ECSDR-BB algorithm performs significantly better than the state-of-the-art SD algorithm for solving the MIMO detection problem in the hard cases (where the number of inputs and outputs is equal or the SNR is low) and the state-of-the-art general-purpose global solver Baron for solving the virtual beamforming design problem.

\appendices

\section{Proof of Proposition~\ref{lemma-width}}\label{applemmawidth}
%\begin{proof}
\textbf{Proof of \eqref{eq10}}: By the definition of $\mathcal{G}_{\mathcal{A}_i}$ (cf. \eqref{GAi1}), we only need to show \begin{equation}\label{rightterm}\left|x_i\right|\geq r_i\cos\left(\frac{\bar{\theta}_i - \underline{\theta}_i}{2}\right).\end{equation}
%we have $r_i\geq |x_i|$
%for all $(x_i,r_i)\in\mathcal{G}_{\mathcal{A}_i}$.

Let us first consider the special case where $\mathcal{A}_{i} = [\underline{\theta}_i,\bar{\theta}_i]$. Without loss of generality, let us assume
$r_i>0$ in \eqref{rightterm}. %and $\bar{\theta}_i - \underline{\theta}_i<\pi$
(Otherwise, if $r_i=0$, then $|x_i|\leq r_i=0$ holds and thus \eqref{eq10} holds.) %; or if $\bar{\theta}_i - \underline{\theta}_i=\pi$, then \eqref{eq10} also holds.)
Since $(x_i,r_i)\in\mathcal{G}_{\mathcal{A}_i}$, we have
\begin{equation}\label{eq14}
\alpha_i \mathrm{Re}(x_i)+\beta_i \mathrm{Im}(x_i)\geq \gamma_i r_i,
\end{equation}
where $\alpha_i, \beta_i, \text{and}~\gamma_i$ are given in \eqref{alpha}, %and they satisfy
%\begin{equation}\label{eq15}
%a_i^2+b_i^2=1,
%\end{equation}
which further implies
%Moreover, by the Cauchy-Schwarz inequality, we have
\begin{equation}\label{eq16}
\alpha_i \mathrm{Re}(x_i)+\beta_i \mathrm{Im}(x_i)\leq \sqrt{\left[\alpha_i^2+\beta_i^2\right]\left[\mathrm{Re}^2(x_i)+\mathrm{Im}^2(x_i)\right]}=\left|x_i\right|.
\end{equation} Combining \eqref{eq14} and \eqref{eq16} with the definition of $\gamma_i$ shows that \eqref{rightterm} holds for all $(x_i,r_i)\in\mathcal{G}_{[\underline{\theta}_i,\bar{\theta}_i]}$.
%Using equations \eqref{eq14} and \eqref{eq16}, we can derive that
%$|x_i|\geq r_i\cos\frac{\bar{\theta}_i - \underline{\theta}_i}{2}$ for all $(x_i,r_i)\in\mathcal{G}_{[\underline{\theta}_i,\bar{\theta}_i]}$.

Now let us consider the general case where $\mathcal{A}_{i}$ is any set satisfying $\mathcal{A}_{i} \subseteq [\underline{\theta}_i,\bar{\theta}_i]$. Since the convex envelope of $\mathcal{A}_{i}$ must also be a subset of that of $[\underline{\theta}_i,\bar{\theta}_i],$ we have $\mathcal{G}_{\mathcal{A}_i} \subseteq \mathcal{G}_{ [\underline{\theta}_i,\bar{\theta}_i]}.$ From this and the result for the case where $\mathcal{A}_{i} = [\underline{\theta}_i,\bar{\theta}_i],$ we can conclude that \eqref{rightterm} holds for all $(x_i,r_i)\in\mathcal{G}_{\mathcal{A}_i}$.
%\end{proof}

%\subsection{Proof of Lemma 2}
%\begin{proof}
\textbf{Proof of \eqref{eq11}}: By the definition of $\mathcal{F}_{\mathcal{B}_i}$ (cf. \eqref{FBi}), we only need to show the second inequality in \eqref{eq11}.
% that $X_{ii}-r_i^2\leq \frac{\left(u_i-\ell_i \right)^2}{4}$.
Since $X_{ii} -(\ell_i+u_i)r_i+\ell_i u_i\leq 0$ holds for all $(X_{ii},r_i)\in\mathcal{F}_{\mathcal{B}_i}$,
it follows
\begin{align*}
X_{ii}-r_i^2 &\leq (\ell_i+u_i)r_i-\ell_i u_i -r_i^2 \\
&= \frac{\left(u_i-\ell_i \right)^2}{4} -\left(r_i- \frac{\ell_i+u_i}{2}\right)^2\\
& \leq \frac{\left(u_i-\ell_i \right)^2}{4}.
\end{align*}
%\end{proof}

\section{Proof of Theorem 2} \label{appthmepssolution}
%\begin{proof}
At the beginning of the $k$-th iteration of the ECSDR-BB algorithm, the initial feasible set $\mathcal{D}^0$ has been (recursively) partitioned into $k$ smaller subsets.
Since the global solution of problem (CQP) must be lie in one of the subsets and $L^k$ is the smallest lower bound of all subproblems in the active node set $\mathcal{P}$, we have $L^k\leq \nu^*$. Since $\hat{\bx}^k=\mathrm{Scale}\left(\bx^k,\br^k\right)$ (cf. \eqref{scale}) is feasible to problem (CQP), we immediately get $\nu^* \leq F\left(\hat{\bx}^k\right)$. Combining the above two inequalities yields \eqref{lowerupper}.

Next, we prove that the returned solution $\bx^*$ by the ECSDR-BB algorithm is an $\epsilon$-optimal solution. It follows from \eqref{rgap} and \eqref{lowerupper} that $F(\hat{\bx}^k) \leq L^k+ \epsilon \leq \nu^*+ \epsilon.$
By the update rule of the upper bound (cf. {Line \ref{line:upper1}} of the ECSDR-BB algorithm), $U^*$ must satisfy $U^*\leq F(\hat{\bx}^k)$. Hence,
for the returned solution $\bx^*$, there holds
$F(\bx^*)=U^* \leq \nu^*+ \epsilon,$ which, together with Definition \ref{esolution}, shows that $\bx^*$ is an $\epsilon$-optimal solution of problem (CQP).  %[[[\rev{There are ??? in the above!!!}]]]
%Then, the termination rule in Line 9 is satisfied,

%\end{proof}

\section{Proof of Lemma \ref{lemmabound}} \label{applemmabound}
%\begin{proof}
By the definition of $F(\bx)$ in problem (CQP), we have
\begin{align*}
&~F\left(\hat{\bx}^k\right)-L^k\\
=&~F\left(\hat{\bx}^k\right)-\frac{1}{2}\bQ\bullet \bX^k-\mathrm{Re}\left(\bc^\dag \bx^k\right) \\
 \leq &~\left|F\left(\hat{\bx}^k\right)-F\left(\bx^k\right)\right|+\left|F\left(\bx^k\right)-\frac{1}{2}\bQ\bullet \bX^k-\mathrm{Re}\left(\bc^\dag \bx^k\right)\right| \\
=&~\left|F\left(\hat{\bx}^k\right)-F\left(\bx^k\right)\right|+\left|\frac{1}{2}\bQ\bullet\left(\bX^k-\bx^k \left(\bx^k\right)^\dag\right)\right|.
\end{align*}
Next, we bound the two terms $\left|F\left(\hat{\bx}^k\right)-F\left(\bx^k\right)\right|$ and $\left|\bQ\bullet\left(\bX^k-\bx^k \left(\bx^k\right)^\dag\right)\right|$ from the above one by one.

We first bound the term $\left|F\left(\hat{\bx}^k\right)-F\left(\bx^k\right)\right|$. %Note that since
%$$\left|F(\bx)-F(\bx')\right|\leq M_F \left\|\bx-\bx'\right\|,~~ \forall \bx,\bx'\in \mathcal{H},$$
%let $(\bx,\bx')=(\bx^k,\hat{\bx}^k)$,
It follows directly from \eqref{eq21} that
\begin{equation}\label{eq21b}
\left|F(\bx^k)-F(\hat{\bx}^k)\right|\leq M_F \left\|\bx^k-\hat{\bx}^k \right\|_2.
\end{equation}
%Besides, %under the condition that $\mathcal{A}_{i_1^*}=[\underline{\theta}_{i_1^*},\bar{\theta}_{i_1^*}]$ is an interval with
%$\bar{\theta}_{i_1^*}-\underline{\theta}_{i_1^*}\leq \pi$, we have
%\begin{equation}
%\hat{\bx}=\textrm{Scale}(\bar{\bx},\bar{\br})=[\bar{r}_1 e^{\textbf{i}\bar{\theta}_1},\ldots,\bar{r}_n e^{\textbf{i}\bar{\theta}_n}]^{\T}
%\end{equation}
%with $\bar{\theta}_{i_1^*} =\arg \left(x_{i_1^*}\right)$.
%Then,
By the definition of $S_1^*$ (cf. \eqref{eq18}), we immediately get
\begin{equation}\label{eq22}
\left\|\bx^k-\hat{\bx}^k \right\|_2%=\sqrt{\sum_{i=1}^{n}(\bar{x}_i-\hat{x})^2}
%=\sqrt{\sum_{i=1}^{n}\left(x^k_i-\hat{x}^k_i\right)^2}
\leq \sqrt{n}S_1^*.
\end{equation}
Combining \eqref{eq21b} and \eqref{eq22} gives
\begin{equation}\label{eq23}
\left|F(\bx^k)-F(\hat{\bx}^k)\right|\leq \sqrt{n}M_F S_1^*.
\end{equation}
%Let $M_a=\sqrt{n}M_F$, then $|F(\bar{x})-F(\hat{x})|\leq M_a (\bar{r}_{i_1^*}-|\bar{x}_{i_1^*}|)$.

Now, we bound the term $\left|\bQ\bullet \left(\bX^k-\bx^k \left(\bx^k\right)^\dag\right)\right|.$
%By using the Neumann inequality, we get
Clearly, %we have $$\left|\bQ\cdot \left(\bX^k-\bx^k \left(\bx^k\right)^\dag\right)\right| \leq $$
there holds
\begin{equation}\label{eq24}\begin{array}{cl}
&\left|\bQ\bullet \left(\bX^k-\bx^k \left(\bx^k\right)^\dag\right)\right| \\ \leq &\left\|\bQ\right\|_{\mathrm{F}} \left\|\left(\bX^k-\bx^k \left(\bx^k\right)^\dag\right)\right\|_{\mathrm{F}}.
\end{array}\end{equation}
%and
%\begin{equation}\label{eq25}
%\left\|\left(\bX^k-\bx^k \left(\bx^k\right)^\dag\right)\right\|=\sqrt{\sum_{i=1}^n \mu_i^2}\leq \sqrt{n} \mu_n,
%\end{equation}
%where $\mu_1,\mu_2, \ldots ,\mu_n$ are eigenvalues
%of $\left(\bX^k-\bx^k \left(\bx^k\right)^\dag\right)$, and are sorted to $0\leq\mu_1\leq \mu_2 \leq \ldots \leq \mu_n$. %, with $0\leq\mu_1\leq \mu_2 \leq ... \leq \mu_n$.
Let $\lambda_{\max}\geq 0$ be the largest eigenvalue of the positive semidefinite matrix $\bX^k-\bx^k \left(\bx^k\right)^\dag.$ Then, we have
\begin{equation}\begin{array}{cl}
&\left\|\bX^k-\bx^k \left(\bx^k\right)^\dag\right\|_{\mathrm{F}} \\ \leq&~\sqrt{n} \lambda_{\max} \leq~\sqrt{n} \textrm{Trace} \left(\bX^k-\bx^k \left(\bx^k\right)^\dag\right).\label{csinequality}
%=&\sqrt{n} \sum_{i=1}^n \left(X^k_{ii}-\left|x^k_i\right|^2\right)\\
%=&\sqrt{n} \sum_{i=1}^n \left[\left(X^k_{ii}-\left(r^k_i\right)^2\right)+\left(r^k_i+\left|x^k_i\right|\right)\left(r^k_i-\left|x^k_i\right|\right)\right]\\
%\leq &n^{\frac{3}{2}} \left(X^k_{i_2^{*}i_2^{*}}-\left(r^k_{i_2^{*}}\right)^2\right)+  2  n^{\frac{3}{2}}u_{\max}\left(r^k_{i_1^*}-\left|x^k_{i_1^*}\right|\right)\\
%=&n^{\frac{3}{2}} S_2^{*}+  2  n^{\frac{3}{2}}u_{\max} S_1^{*},
\end{array}\end{equation} By the definitions of $S_1^{*}$ and $S_2^{*}$ (cf. \eqref{eq18}), we have
\begin{align*}
%&\left\|\bX^k-\bx^k \left(\bx^k\right)^\dag\right\| \leq \sqrt{n} \lambda_{\max} \leq \sqrt{n}
&~\textrm{Trace} \left(\bX^k-\bx^k \left(\bx^k\right)^\dag\right)\nonumber\\
=&~\sum_{i=1}^n \left(X^k_{ii}-\left|x^k_i\right|^2\right)\nonumber\\
=&~\sum_{i=1}^n \left[\left(X^k_{ii}-\left(r^k_i\right)^2\right)+\left(r^k_i+\left|x^k_i\right|\right)\left(r^k_i-\left|x^k_i\right|\right)\right]\label{eqtrace}\\
\leq &~n \left[\left(X^k_{i_2^{*}i_2^{*}}-\left(r^k_{i_2^{*}}\right)^2\right)+  2  u_{{\max}}\left(r^k_{i_1^*}-\left|x^k_{i_1^*}\right|\right)\right]\nonumber\\
=&~n\left( S_2^{*}+  2 u_{\max} S_1^{*}\right),\nonumber
\end{align*} which, together with \eqref{eq24} and \eqref{csinequality}, further implies
\begin{equation}\label{eqsecondterm}
\left|\bQ\bullet \left(\bX^k-\bx^k \left(\bx^k\right)^\dag\right)\right|\leq \|\bQ\|_{\mathrm{F}} n^{\frac{3}{2}} \left( S_2^{*}+  2 u_{\max} S_1^{*}\right).
\end{equation} From \eqref{eq23}, \eqref{eqsecondterm}, and the definitions of $M_1$ and $M_2$ (cf. \eqref{eq27} and \eqref{eq27c}), we immediately get the desired inequality in \eqref{eq20}.
%
%Finally, let $M_1$ and $M_2$ be the constants defined by \eqref{eq27} and \eqref{eq27c},
%equation %\eqref{eq20}
%$$\left|F(\hat{\bx}^k)-L^k\right|\leq  M_1 S_1^*+M_2 S_2^*$$ is derived.
%\end{proof}

\section{Proof of Lemma \ref{lemmaterminate}}\label{applemmaterminate}
%\begin{proof}
It follows from Theorem \ref{thm-epssolution} that, to prove the lemma we only need to prove that \eqref{rgap} holds under (C1), (C2), or (C3).

If condition (C1) holds, then it follows from \eqref{eq20} that \begin{equation}\label{eq29}
F(\hat{\bx}^k)-L^k\leq \left(M_1+M_2\right) S^*_1.
\end{equation}
In this case, we have $\hat{x}^k_{i_1^*}=r^k_{i_1^*} e^{\textsf{i} \arg (x^k_{i_1^*})}$ (cf. \eqref{scale})
and thus \begin{equation}\label{abserr}\left|\hat{x}^k_{i_1^*}-x^k_{i_1^*}\right|=r^k_{i_1^*}-\left|x^k_{i_1^*}\right|.\end{equation} Then, we have
\begin{align}
  S^*_1 =  \left|\hat{x}^k_{i_1^*}-x^k_{i_1^*}\right|        = &~r^k_{i_1^*}-\left|x^k_{i_1^*}\right|\nonumber\\
        \leq &~{r}_{i_1^*}^k \left[1-\cos \left(\frac{ \bar{\theta}^k_{i_1^*}- \underline{\theta}^k_{i_1^*} }{2}\right)\right]\label{S1}\\
        \leq &~\frac{u_{\max}\left(\bar{\theta}^k_{i_1^*}-\underline{\theta}^k_{i_1^*}\right)^2}{8}\leq  \frac{u_{\max}\kappa_1^2}{8},\nonumber
        %\leq &~u_{\max} \left(1-\cos \frac{ \bar{\theta}^k_{i_1^*}- \underline{\theta}^k_{i_1^*} }{2}\right),
\end{align} where the first equality is due to the definition of $S_{1}^*$ (cf. \eqref{eq18}),
the second equality comes from \eqref{abserr}, the first inequality is due to \eqref{eq10}, and the second inequality is a result of the definition of $u_{\max}$ (cf. \eqref{umax}) and the
inequality $1-\cos\left(\theta\right) \leq \frac{\theta^2}{2}$ for all $\theta\in\mathbb{R},$ and the last inequality follows from condition (C1).
Combining \eqref{eq29}, \eqref{S1}, and the definition of $\kappa_1$ (cf. \eqref{eq32}) yields the desired result in \eqref{rgap}.

%Based on Lemma 1, we have
%\begin{equation}\label{eq30}
%S^*_1\leq \bar{r}_{i_1^*} \left(1-\cos \frac{ \bar{\theta}^k_{i_1^*}- \underline{\theta}^k_{i_1^*} }{2}\right) \leq u_{\max} \left(1-\cos \frac{ \bar{\theta}^k_{i_1^*}- \underline{\theta}^k_{i_1^*} }{2}\right).
%\end{equation}
%
%
%%Since $L^k=\frac{1}{2}\bQ\cdot \bX^k+\mathrm{Re}(\bc^\dag \bx^k)$,
%Based on Lemma 3, we have
%\begin{equation}\label{eq28}
%F(\hat{\bx}^k)-L^k \leq  M_1 S_1^*+M_2 S_2^*.
%\end{equation}
%If condition (C1) holds, then
%
%
%Also note that the inequality $1-\cos \phi \le \frac{\phi^2}{2}$ holds for any $\phi \in \mathbb{R}$, and apply this
%inequality to the right side of \eqref{eq30}, we derive
%\begin{equation}\label{eq31}
%S^*_1 \leq  \frac{u_{\max}\left(\bar{\theta}^k_{i_1^*}-\underline{\theta}^k_{i_1^*}\right)^2}{8}.
%\end{equation}
%Then, based on equations \eqref{eq29} and \eqref{eq31}, and the definition of $\delta_1$ in \eqref{eq32}, we can derive
%that $F\left(\hat{\bx}^k\right)-L^k\leq \epsilon$ if condition (C1) holds.

If condition (C2) holds, then we can show $S_1^*=S_2^*=0.$ By this and \eqref{eq20}, we obtain
$F(\hat{\bx}^k)-L^k \leq 0,$ which further implies \eqref{rgap}.

If condition (C3) holds, then it follows from \eqref{eq20} that
\begin{equation}\label{eq33}
F\left(\hat{\bx}^k\right)-L^k \leq  \left(M_1+M_2\right)S_2^*.
\end{equation}
Moreover, from \eqref{eq11} and the definition of $S_2^*$ (cf. \eqref{eq18}), we obtain
%\begin{equation}\label{eq34}
$S^*_2 = X_{i_2^*i_2^*} - r_{i_2^*}^2\leq {\left(u^k_{i_2^*}-\ell^k_{i_2^*}\right)^2}/{4}.$
%\end{equation}
This, together with \eqref{eq33}, (C3), and the definition of $\kappa_2$ (cf. \eqref{eq35}), shows the desired result in \eqref{rgap}.
%\end{proof}

\section{Proof of Theorem \ref{thm-iteration}}\label{appthmiteration}
%\begin{proof}
%For any given index $i\in\{1,2,\ldots,n\}$,
We consider two sets $\mathcal{A}_i$ and $\mathcal{B}_i$ separately. Moreover, when we consider set $\mathcal{A}_i,$ we consider two cases where $\mathcal{A}_i$ is an interval and a discrete set separately.

We first consider the case where $\mathcal{A}_i$ is an interval.
We show that the set $\mathcal{A}_i$ will be partitioned into at most $\mu(\mathcal{A}_i)$ of subsets before the algorithm terminates, where $\mu(\mathcal{A}_i)$ is defined  in \eqref{muA}. According to the algorithm, suppose that %the selected index $i^*_1$ is $i$ and
$S_1^*\geq S_2^*$ at the $k$-th iteration,
%For the , if the selected index  and  %and $\mathcal{A}^k_{i_1^*}$ is an interval set,
then the interval $\mathcal{A}^k_{i_1^*}$
will be partitioned into two subsets with the same length. If the ECSDR-BB algorithm does not terminate in Line \ref{line:terminate} at the $k$-th iteration, then it follows from condition (C1) in Lemma \ref{lemmaterminate} that $\omega(\mathcal{A}^k_{i_1^*})>\min\left\{\kappa_1,\pi\right\}$ and the length of each subset obtained after the partition
is larger than $\frac{1}{2} \min\{\kappa_1,\pi\}$. {Hence, if set $\mathcal{A}_i$ has been partitioned into $\mu(\mathcal{A}_i)$ of subsets, the total length of all obtained subsets is strictly greater than  $$\mu(\mathcal{A}_i)\frac{1}{2} \min\{\kappa_1,\pi\}\geq \omega(\mathcal{A}_i),$$ where the inequality is due to the definition of $\mu(\mathcal{A}_i)$ (cf. \eqref{muA}). This is a contradiction.
Therefore, if $\mathcal{A}_i$ is an interval,
it can be partitioned at most $\mu(\mathcal{A}_i)$
%$$\max\left\{\left\lceil \frac{2\omega(\mathcal{A}_i)}{\min\{\delta_1,\pi\}}\right\rceil,1\right\}$$
times before the algorithm terminates.}

Now, we consider the case where $\mathcal{A}_i$ is a discrete set (with a finite number of elements). We can use the similar argument as in the above case to show that $\mathcal{A}_i$ can be partitioned at most $\mu(\mathcal{A}_i)=\left|\mathcal{A}_i\right|$ times. The only difference here is that $\mathcal{A}_i$ is a discrete set.
More specifically, according to the algorithm, suppose that %the selected index $i^*_1$ is $i$ and
$S_1^*\geq S_2^*$ at the $k$-th iteration,
%For the , if the selected index  and  %and $\mathcal{A}^k_{i_1^*}$ is an interval set,
then the interval $\mathcal{A}^k_{i_1^*}$
will be partitioned into two nonempty and nonoverlapping subsets.
%with a balanced cardinality\footnote{Two sets with a balanced cardinality here means that the difference of the cardinality of the two sets is at most one.}.
If the ECSDR-BB algorithm does not terminate in Line \ref{line:terminate} at the $k$-th iteration, then it follows from condition (C2) in Lemma \ref{lemmaterminate} that $\mathcal{A}^k_{i_1^*}$ is not a singleton and each subset obtained after the partition
is not empty. Hence, if $\mathcal{A}_i$ is a discrete set, it can be partitioned at most $\left|\mathcal{A}_i\right|$
%$$\max\left\{\left\lceil \frac{2\omega(\mathcal{A}_i)}{\min\{\delta_1,\pi\}}\right\rceil,1\right\}$$
times before the algorithm terminates.

%Based on Condition (C2) in Lemma 4, we can show that  where $\left|\mathcal{A}_i\right|$ denotes the cardinality of the set $\mathcal{A}_i$.

Finally, we consider set $\mathcal{B}_i$. This case is essentially the same as the case where $\mathcal{A}_i$ is an interval. Using the same argument, we can show that the set $\mathcal{B}_i$ can be partitioned at most
$\max\left\{\left\lceil \frac{2\omega(\mathcal{B}_i)}{\kappa_2}\right\rceil,1\right\}$ times before the algorithm terminates.
%For the $k$-th iteration, if the selected index $i^*_2$ equals $i$ and $S_1^*< S_2^*$, then the interval $\mathcal{B}_{i_2^*}=\left[\ell^k_{i_2^*},u^k_{i_2^*}\right]$ will be
%partitioned to two sub-intervals, and the length of each sub-interval is larger than $\frac{\delta_2}{2}$.
%Thus, the set $\mathcal{B}_i$ can be partitioned to at most
%$$\max\left\{\left\lceil \frac{2\omega(\mathcal{B}_i)}{\delta_2}\right\rceil,1\right\}$$ sub-intervals
%before the algorithm terminates.

From the above analysis, we can conclude that the proposed algorithm must terminate within at most $K_\epsilon$ iterations, where $K_\epsilon$ is defined in \eqref{eq36}.
%$$
%\prod_{i=1}^n  N(\mathcal{A}_i) \cdot \max\left\{\left\lceil \frac{2\omega(\mathcal{B}_i)}{\delta_2}\right\rceil,1\right\}
%$$
%iterations.

%\end{proof}

%[[[\rev{I believe that the proof is correct. However, we have not made it very clear. For instance, why there is a $2$ in the numerators of \eqref{eq36} and \eqref{muA}? Let us further work on this! }]]]

\ifCLASSOPTIONcaptionsoff
  \newpage
\fi

%[[[\rev{References need to be carefully updated!}]]]

%[[[\rev{Add the conference version of the paper!}]]]

%\newpage

\bibliographystyle{IEEEtran}
%\bibliography{PC-BB-11_22}
%\bibliography{ECSDR-BB-Ref-1124}
%\bibliography{ECSDR-BB-Ref-20190904}
%\bibliography{ECSDR-BB-Ref-20200524}
\bibliography{ECSDR-BB-Ref-20200525}

\begin{thebibliography}{10}
\providecommand{\url}[1]{#1}
\csname url@samestyle\endcsname
\providecommand{\newblock}{\relax}
\providecommand{\bibinfo}[2]{#2}
\providecommand{\BIBentrySTDinterwordspacing}{\spaceskip=0pt\relax}
\providecommand{\BIBentryALTinterwordstretchfactor}{4}
\providecommand{\BIBentryALTinterwordspacing}{\spaceskip=\fontdimen2\font plus
\BIBentryALTinterwordstretchfactor\fontdimen3\font minus
  \fontdimen4\font\relax}
\providecommand{\BIBforeignlanguage}[2]{{%
\expandafter\ifx\csname l@#1\endcsname\relax
\typeout{** WARNING: IEEEtran.bst: No hyphenation pattern has been}%
\typeout{** loaded for the language `#1'. Using the pattern for}%
\typeout{** the default language instead.}%
\else
\language=\csname l@#1\endcsname
\fi
#2}}
\providecommand{\BIBdecl}{\relax}
\BIBdecl

\bibitem{luefficienticcc}
C.~Lu, Y.-F. Liu, and J.~Zhou, ``An efficient global algorithm for nonconvex
  complex quadratic problems with applications in wireless communications,'' in
  \emph{Proc. IEEE/CIC Int. Conf. Commun. China (ICCC)}, Oct. 2017, pp. 1--5.

\bibitem{Jalden}
J.~Jald\'{e}n, C.~Martin, and B.~Ottersten, ``Semidefinite programming for
  detection in linear systems--{O}ptimality conditions and space-time
  decoding,'' in \emph{Proc. IEEE Int. Conf. Acoust. Speech Signal Process.
  (ICASSP)}, Apr. 2003, pp. 9--12.

\bibitem{Ma2004}
W.-K. Ma, P.-C. Ching, and Z.~Ding, ``Semidefinite relaxation based multiuser
  detection for {$M$-ary} {PSK} multiuser systems,'' \emph{IEEE Trans. Signal
  Process.}, vol.~52, no.~10, pp. 2862--2872, Oct. 2004.

\bibitem{Maio2009}
A.~D. Maio, S.~D. Nicola, Y.~Huang, Z.-Q. Luo, and S.~Zhang, ``Design of phase
  codes for radar performance optimization with a similarity constraint,''
  \emph{IEEE Trans. Signal Process.}, vol.~57, no.~2, pp. 610--621, Feb. 2009.

\bibitem{Maio2011}
A.~D. Maio, Y.~Huang, M.~Piezzo, S.~Zhang, and A.~Farina, ``Design of optimized
  radar codes with a peak to average power ratio constraint,'' \emph{IEEE
  Trans. Signal Process.}, vol.~59, no.~6, pp. 2683--2697, Jun. 2011.

\bibitem{Soltanalian}
M.~Soltanalian and P.~Stoica, ``Designing unimodular codes via quadratic
  optimization,'' \emph{IEEE Trans. Signal Process.}, vol.~62, no.~5, pp.
  1221--1234, Mar. 2014.

\bibitem{Hong}
M.~Hong, Z.~Xu, M.~Razaviyayn, and Z.-Q. Luo, ``Joint user grouping and linear
  virtual beamforming: {C}omplexity, algorithms and approximation bounds,''
  \emph{IEEE J. Sel. Areas Commun.}, vol.~31, no.~10, pp. 2013--2027, Oct.
  2013.

\bibitem{Waldspurger}
I.~Waldspurger, A.~d'Aspremont, and S.~Mallat, ``Phase recovery, {MaxCut} and
  complex semidefinite programming,'' \emph{Math. Program.}, vol. 149, no.
  1--2, pp. 47--81, Feb. 2015.

\bibitem{pu2018optimal}
W.~Pu, Y.-F. Liu, J.~Yan, H.~Liu, and Z.-Q. Luo, ``Optimal estimation of sensor
  biases for asynchronous multi-sensor data fusion,'' \emph{Math. Program.},
  vol. 170, no.~1, pp. 357--386, 2018.

\bibitem{Bandeira}
A.~S. Bandeira, N.~Boumal, and A.~Singer, ``Tightness of the maximum likelihood
  semidefinite relaxation for angular synchronization,'' \emph{Math. Program.},
  vol. 163, no. 1-2, pp. 145--167, May 2017.

\bibitem{Luo2010}
Z.-Q. Luo, W.-K. Ma, A.~M.-C. So, Y.~Ye, and S.~Zhang, ``Semidefinite
  relaxation of quadratic optimization problems,'' \emph{IEEE Signal Process.
  Mag.}, vol.~27, no.~3, pp. 20--34, May 2010.

\bibitem{Palomar}
D.~P. Palomar and Y.~C. Eldar, \emph{Convex Optimization in Signal Processing
  and Communications}.\hskip 1em plus 0.5em minus 0.4em\relax New York, USA:
  Cambridge University Press, 2010.

\bibitem{So2008}
A.~M.-C. So, J.~Zhang, and Y.~Ye, ``On approximating complex quadratic
  optimization problems via semidefinite programming relaxations,'' \emph{Math.
  Program.}, vol. 110, no.~1, pp. 93--110, Jun. 2007.

\bibitem{Zhang2006}
S.~Zhang and Y.~Huang, ``Complex quadratic optimization and semidefinite
  programming,'' \emph{SIAM J. Optim.}, vol.~16, no.~3, pp. 871--890, 2006.

\bibitem{Lu2018}
C.~Lu, Z.~Deng, W.-Q. Zhang, and S.-C. Fang, ``Argument division based
  branch-and-bound algorithm for unit-modulus constrained complex quadratic
  programming,'' \emph{J. Global Optim.}, vol.~70, no.~1, pp. 171--187, Jan.
  2018.

\bibitem{Goemans1995}
M.~X. Goemans and D.~P. Williamson, ``Improved approximation algorithms for
  maximum cut and satisfiability problems using semidefinite programming,''
  \emph{J. ACM}, vol.~42, no.~6, pp. 1115--1145, Nov. 1995.

\bibitem{Goemans2004}
------, ``Approximation algorithms for {Max-3-Cut} and other problems via
  complex semidefinite programming,'' \emph{J. Comput. Syst. Sci.}, vol.~68,
  no.~2, pp. 442--470, Mar. 2004.

\bibitem{JarreJoGO}
F.~Jarre, F.~Lieder, Y.-F. Liu, and C.~Lu, ``Set-completely-positive
  representations and cuts for the max-cut polytope and the unit modulus
  lifting,'' \emph{J. Global Optim.}, vol.~76, no.~4, pp. 913--932, 2020.

\bibitem{Gershman}
A.~B. Gershman, N.~D. Sidiropoulos, S.~Shahbazpanahi, M.~Bengtsson, and
  B.~Ottersten, ``Convex optimization-based beamforming,'' \emph{IEEE Signal
  Process. Mag.}, vol.~27, no.~3, pp. 62--75, May 2010.

\bibitem{He}
S.~He, Z.-Q. Luo, J.~Nie, and S.~Zhang, ``Semidefinite relaxation bounds for
  indefinite homogeneous quadratic optimization,'' \emph{SIAM J. Optim.},
  vol.~19, no.~2, pp. 503--523, 2008.

\bibitem{Luo2007}
Z.-Q. Luo, N.~D. Sidiropoulos, P.~Tseng, and S.~Zhang, ``Approximation bounds
  for quadratic optimization with homogeneous quadratic constraints,''
  \emph{SIAM J. Optim.}, vol.~18, no.~1, pp. 1--28, 2007.

\bibitem{Ge2015}
R.~Ge, F.~Huang, C.~Jin, and Y.~Yuan, ``Escaping from saddle points -- online
  stochastic gradient for tensor decomposition,'' in \emph{Proc. The 28th
  Conference on Learning Theory (COLT)}, Jul. 2015, pp. 797--842.

\bibitem{lu2019tightness}
C.~Lu, Y.-F. Liu, W.-Q. Zhang, and S.~Zhang, ``Tightness of a new and enhanced
  semidefinite relaxation for {MIMO} detection,'' \emph{SIAM J. Optim.},
  vol.~29, no.~1, pp. 719--742, 2019.

\bibitem{Linderoth}
J.~Linderoth, ``A simplicial branch-and-bound algorithm for solving
  quadratically constrained quadratic programs,'' \emph{Math. Program.}, vol.
  103, no.~2, pp. 251--282, Jun. 2005.

\bibitem{Tawarmalani}
M.~Tawarmalani and N.~V. Sahinidis, ``A polyhedral branch-and-cut approach to
  global optimization,'' \emph{Math. Program.}, vol. 103, no.~2, pp. 225--249,
  Jun. 2005.

\bibitem{Chan}
A.~M. Chan and I.~Lee, ``A new reduced-complexity sphere decoder for multiple
  antenna systems,'' in \emph{Proc. IEEE Int. Conf. Commun. (ICC)}, Apr.-May
  2002, pp. 460--464.

\bibitem{Damen}
O.~Damen, A.~Chkeif, and J.-C. Belfiore, ``Lattice code decoder for space-time
  codes,'' \emph{IEEE Commun. Lett.}, vol.~4, no.~5, pp. 161--163, May 2000.

\bibitem{BenTal}
A.~Ben-Tal and A.~Nemirovski, \emph{Lectures on Modern Convex
  Optimization}.\hskip 1em plus 0.5em minus 0.4em\relax Philadelphia, PA, USA:
  SIAM, 2001.

\bibitem{lu2017efficient}
C.~Lu and Y.-F. Liu, ``An efficient global algorithm for single-group multicast
  beamforming,'' \emph{IEEE Trans. Signal Process.}, vol.~65, no.~14, pp.
  3761--3774, Jul. 2017.

\bibitem{Sedumi}
J.~F. Sturm, ``Using {SeDuMi} 1.02, {A} {MATLAB} toolbox for optimization over
  symmetric cones,'' \emph{Optim. Methods Softw.}, vol.~11, no. 1--4, pp.
  625--653, 1999.

\end{thebibliography}


%\end{IEEEbiography}

%\newpage

%\begin{IEEEbiographynophoto}{Cheng Lu}


%\vfill


\end{document}